\documentclass[journal, onecolumn, 11pt]{IEEEtran}

\usepackage{setspace}
\usepackage{amsmath}
\usepackage{amssymb}
\usepackage{mathrsfs}
\usepackage{amsthm}
\usepackage{multirow}
\usepackage{color}
\usepackage{array}
\newcolumntype{C}[1]{>{\centering\arraybackslash}p{#1}}
\usepackage{url}
\usepackage{comment}
\usepackage{enumerate}
\usepackage{eucal}
\usepackage{hyperref}
\newcommand\myshade{50}

\hypersetup{
	linkcolor  = red!\myshade!black,
	citecolor  = blue!\myshade!black,
	urlcolor   = blue!\myshade!black,
	colorlinks = true,
}

\usepackage{subcaption}

\usepackage{algorithm}
\usepackage[noend]{algorithmic}
\usepackage{cite}
\usepackage[table]{xcolor}
\usepackage{blkarray}

\usepackage{stmaryrd}

\usepackage[top=1in, bottom=1in, left=0.75in, right=0.75in]{geometry}

\newcommand{\ra}[1]{\renewcommand{\arraystretch}{#1}}

\usepackage{tikz}
\usetikzlibrary{arrows, patterns, shapes.arrows, decorations.pathmorphing, matrix, calc, decorations.pathreplacing, arrows.meta,automata,positioning}

\pgfkeys{tikz/mymatrixenv/.style={decoration=brace,every left delimiter/.style={xshift=3pt},every right delimiter/.style={xshift=-3pt}}}
\pgfkeys{tikz/mymatrix/.style={matrix of math nodes,left delimiter=[,right delimiter={]},inner sep=2pt,column sep=1em,row sep=0.5em,nodes={inner sep=0pt}}}
\pgfkeys{tikz/mymatrixbrace/.style={decorate,thick}}

\IEEEoverridecommandlockouts

\theoremstyle{plain}
\newtheorem{theorem}{Theorem}
\newtheorem{corollary}[theorem]{Corollary}
\newtheorem{lemma}[theorem]{Lemma}
\newtheorem{proposition}[theorem]{Proposition}

\theoremstyle{definition}

\newtheorem{example}[theorem]{Example}

\newtheorem{remark}[theorem]{Remark}

\newcommand{\RR}{\mathbb{R}}

\newcommand{\HH}{\mathbb{H}}

\DeclareMathAlphabet{\mathbfsl}{OT1}{ppl}{b}{it} 

\newcommand{\ve}{\mathbfsl{e}}

\newcommand{\va}{\mathbfsl{a}}

\newcommand{\vq}{\mathbfsl{q}}
\newcommand{\vr}{\mathbfsl{r}}
\newcommand{\vs}{\mathbfsl{s}}

\newcommand{\vx}{\mathbfsl{x}}
\newcommand{\vy}{\mathbfsl{y}}

\newcommand{\vA}{\mathbfsl{A}}
\newcommand{\vB}{\mathbfsl{B}}
\newcommand{\vC}{\mathbfsl{C}}
\newcommand{\vD}{\mathbfsl{D}}

\newcommand{\vH}{\mathbfsl{H}}

\newcommand{\vM}{\mathbfsl{M}}
\newcommand{\vT}{\mathbfsl{T}}

\newcommand{\cS}{\mathcal{S}}

\newcommand{\cC}{\mathcal{C}}

\newcommand{\cE}{\mathcal{E}}

\newcommand{\cG}{\mathcal{G}}
\newcommand{\cH}{\mathcal{H}}
\newcommand{\cL}{\mathcal{L}}

\newcommand{\cV}{\mathcal{V}}

\newcommand{\cP}{\mathcal{P}}

\newcommand{\rateT}{{\stackrel{\sim}{\smash{T}\rule{0pt}{1.1ex}}}}

\newcommand{\capy}{{\rm Cap}}

\newcommand{\rgv}{R_{\rm GV}}
\newcommand{\rmr}{R_{\rm MR}}
\newcommand{\gvcurve}{{\cal GV}}

\newcommand{\etal}{{\em et al.}}

\newcommand{\floor}[1]{{\left\lfloor #1\right\rfloor}}

\newcommand{\interval}[1]{{\left[  #1\right]}}

\newcommand{\todo}[1]{{\color{red}(TODO: #1)}}

\title{Evaluating the Gilbert-Varshamov Bound for Constrained Systems}

\author{
	\IEEEauthorblockN{
		Keshav Goyal
		and
		Han Mao Kiah\\[2mm]
	}
	\IEEEauthorblockA{\small School of Physical and Mathematical Sciences, Nanyang Technological University, Singapore} \\[0mm]
	\IEEEauthorblockA{\small Emails: 
			keshav002@ntu.edu.sg, hmkiah@ntu.edu.sg\\[0mm]}
	 \thanks{The paper was presented in part at the 2022 IEEE International Symposium on Information Theory (ISIT) \cite{keshav2022evaluating}.}}

\begin{document}
	\date{}
	
	\maketitle
	
	\hspace*{-10pt}\begin{abstract}
	We revisit the well-known Gilbert-Varshamov (GV) bound for constrained systems. 
	In 1991, Kolesnik and Krachkovsky showed that GV bound can be determined via the solution of some optimization problem. Later, Marcus and Roth (1992) modified the optimization problem and improved the GV bound in many instances. In this work, we provide explicit numerical procedures to solve these two optimization problems and hence, compute the bounds.
	We then show the procedures can be further simplified when we plot the respective curves.
	In the case where the graph presentation comprise a single state, we provide explicit formulas for both bounds.
	\end{abstract}
	
\section{Introduction}	


From early applications in magnetic recording systems to recent applications 
in DNA-based data storage \cite{Yazdi2015, Immink2019, Gabrys2020, Nguyen2021, Cai2021,Kovacevic2021} and 
energy-harvesting \cite{Popovski2013,Fouladgar2014,Tandon2016,Immink2020,Immink.2020,Wu2021}, constrained codes play a central role in enhancing reliability in many data storage and communications systems (see also \cite{MRS2001} for a survey).
Specifically, for most data storage systems, certain substrings are more prone to errors than others.
Thus, by forbidding the appearance of such strings, that is, imposing constraints on the codewords, 
the user is able to reduce the likelihood of error.
We refer to the collection of words that satisfy the constraints as the {\em constrained space} $\cS$.

Now, to further reduce the error probability, one can impose certain distance constraints on the codebook.
In this work, we focus on the {\em Hamming metric} and 
consider the maximum size of a codebook whose words belong to the constrained space $\cS$ and whose pairwise distance are at least a certain value $d$. 
Specifically, we study one of the most well-known and fundamental lower bounds to this quantity -- the {\em Gilbert-Varshamov (GV) bound}.

To determine the GV bound, one requires two quantities: the size of the constrained space, $|\cS|$, and also, the {\em ball volume}, that is, the number of words with distance at most $d-1$ from a ``center'' word.
In the case where the space is unconstrained, i.e. $\cS=\{0,1\}^n$, the ball volume does not depend on the center. Then the GV bound is simply $|\cS|/V$ where $V$ is the ball volume for some center.
However, for most constrained systems, the ball volume varies with the center.
Nevertheless, Kolesnik and Krachkovsky showed that the GV lower bound can be generalized to $|\cS|/4\overline{V}$ where $\overline{V}$ is the {\em average ball volume} \cite{Kolesnik1991}.
This was further improved by Gu and Fuja to $|\cS|/\overline{V}$ in ~\cite{GuFuja.1993} (see~\cite[pp. 242-243]{MRS2001} for  additional details).
In the same paper\cite{Kolesnik1991}, they showed the asymptotic rate of average ball volume can be computed via some optimization problem.
Later, Marcus and Roth modified the optimization problem by including an additional constraint and variable~\cite{MarcusRoth1992}, and the resulting bound, which we refer to as {\em GV-MR bound}, improves the usual GV bound. Furthermore, in most cases, the improvement is strictly positive.

However, about three decades later, very few works have evaluated these bounds for specific constrained systems. 
To the best of our knowledge, in all works that computed numerically, the GV bound and / or GV-MR bound, the constrained systems of interest have at most eight states \cite{Winick1996}. 
In~\cite{Winick1996}, the authors wrote that ``evaluation of the bound required considerable computation'', referring to the GV-MR bound.

In this paper, we revisit the optimization problems defined by Kolesnik and Krachkovsky~\cite{Kolesnik1991} and Marcus and Roth~\cite{MarcusRoth1992} and develop a suite of explicit numerical procedures that solve these problems.
In particular, to demonstrate the feasibility of our methods, we evaluate and plot the GV and GV-MR bounds for a constrained system involving 120 states in Fig.~\ref{fig:swcc}(b).

We provide a high-level description of our approach.
For both optimization problems, we first characterize the optimal solutions as roots of certain equations.
Then using the celebrated {\em Newton-Raphson} iterative procedure, we proceed to find the roots of these equations.
However, as the latter equations involve the largest eigenvalues of certain matrices, 
each Newton-Raphson iteration requires the (partial) derivatives of these eigenvalues (in some variables).
To resolve this, we make modifications to another celebrated iterative procedure -- the {\em power iteration} method 
and the resulting procedures compute the GV and GV-MR bounds efficiently for a specific relative distance $\delta$.
Interestingly, if we are plotting the bounds for $0\le \delta\le 1$, the numerical procedure can be further simplified.
Specifically, by exploiting certain properties of the optimal solutions, we provide procedures that use less Newton-Raphson iterations.

In the next section, we provide the formal definitions and state the optimization problems that compute the GV bound.


\section{Preliminaries}
	
Let $\Sigma=\{0,1\}$ be the binary alphabet and let $\Sigma^{n}$ denote the set of all words of length $n$ over $\Sigma$. 
A {\em labelled graph} $\cG = (\cV, \cE, \cL)$ is a 
finite directed graph with {\em states} $\cV$, {\em edges} $\cE \subseteq \cV \times \cV$, and 
an {\em edge labelling} $\cL: \cE \to \Sigma^{s}$ for some $s \geq 1$. 
Here, we use $v_i \xrightarrow{\sigma} v_j$ to mean that there is an edge from $v_i$ to $v_j$ with label $\sigma$. 
The labelled graph $\cG$ is {\em deterministic} if for each state, the outgoing edges have distinct labels.

A {\em constrained system} $\cS$ is then the set of all words obtained by reading the labels of paths in a labelled graph $\cG$. We say that $\cG$ is a {\em graph presentation} of $\cS$. We further denote the set of all length-$n$ words $\cS$ by $\cS_{n}$. 
Alternatively, $\cS_{n}$ is the set of all words obtained by reading the labels of length-$(n/s)$ paths in $\cG$. Then the {\em capacity of $\cS$}, denoted by $\capy(\cS)$ is given by 
$\capy(\cS) \triangleq \limsup_{n\to\infty} {\log |\cS_n|}/{n}$.
It is well-known that $\capy(\cS)$ corresponds to the largest eigenvalue of the {\em adjacency matrix} $\vA_\cG$  (see for example, \cite{MRS2001}).
Here, $\vA_\cG$ is a $(|\cV| \times |\cV|)$-matrix whose rows and columns are indexed by $\cV$. 
For each entry $(u,v)\in \cV\times \cV$, we set the corresponding entry to be one if $(u,v)$ is an edge, and zero, otherwise.

Every constrained system can be presented by a deterministic graph $\cG$ . 
Furthermore, any deterministic graph can be transformed into a primitive deterministic graph $\cH$ such that the capacity of $\cG$ is same as the capacity of the constrained system presented by some irreducible component of $\cH$ (see for example, Marcus \etal{} \cite{MRS2001}).  
Therefore, we henceforth assume that our graphs are deterministic and primitive. 
When $|\cV| = 1$, we call this a {\em single-state graph presentation} and study these graphs in Section~\ref{sec:ss}.

For $\vx,\vy \in \cS$, denoted by $d_{H}(\vx,\vy)$ is the Hamming distance between $\vx$ and $\vy$. 
Fix $1\le d\le n$ and a fundamental problem in coding theory is to find the largest subset $\cC$ of $\cS_{n}$ such that $d_{H}(\vx,\vy) \geq d$ for all distinct $\vx,\vy \in \cC$. Let $A(n,d;\cS)$ denote the size of largest subset $\cC$.

In terms of asymptotic rates, we fix $0 \le  \delta \le 1$ and our task is to find the highest attainable rate, denoted by $R(\delta)$, 
which is given by $R(\delta;\cS) \triangleq \limsup_{n\to\infty} {\log A(n,\floor{\delta n};\cS)}/{n}$.


\subsection{Review of Gilbert-Varshamov Bound}
To define the GV bound, we need to determine the total ball size. Specifically, 
for $\vx\in\cS_n$ and $0\le r\le n$, we define $V(\vx,r;\cS)\triangleq |\{\vy\in\cS_n: d_H(\vx,\vy)\le r\}|$. 
We further define $T(n, d; \cS) = \sum_{\vx\in\cS_n} V(\vx,d-1;\cS)$\,.
Then the GV bound  
as given by Gu and Fuja \cite{GuFuja.1993, Tolhuizhen.1997} states that there exists an $(n,d;\cS)$-code of size at least
${|\cS_n|^2}/{T(n,d;\cS)}$.

In terms of asymptotic rates, there exists a family of $(n,\floor{\delta n};\cS)$-codes such that their rates approach
\vspace{-1mm}
\begin{equation} \label{eq:GV}
	\rgv(\delta) = 2\capy(\cS)-\rateT(\delta),
\end{equation}
where $\rateT(\delta)\triangleq \limsup_{n\to\infty} {\log T(n,\floor{\delta n};\cS)}/{n}$\,.

In this paper, our main task is to determine $\rgv(\delta)$ {\em efficiently}. Observe that since $\capy(\cS)=\rateT(0)$, it suffices to find efficient ways of determining $\rateT(\delta)$.
It turns out that $\rateT(\delta)$ can be found via the solution of some convex optimization problem. 
Specifically, given a labelled graph $\cG= (\cV,\cE,\cL)$, we define its {\em product graph} 
$\cG' = (\cV',\cE',\cL')$ as follows.
\begin{itemize}
	\item  $\cV' \triangleq \cV \times \cV$.
	\item For $(v_i,v_j), (v_k,v_\ell) \in \cV'$ and $(\sigma_1,\sigma_2) \in \Sigma^{s} \times \Sigma^{s}$, we draw an edge $(v_i,v_j) \xrightarrow{(\sigma_1,\sigma_2)} (v_k,v_\ell)$ if and only if both $v_i \xrightarrow{\sigma_1}v_k$ and $v_j \xrightarrow{\sigma_2} v_\ell$ belong to $\cE$.
\end{itemize}

\noindent Then we label the edges in $\cE'$ with the function $D:\cE' \to \mathbb{Z}_{\geq 0}$, where $D\left((v_i,v_j) \xrightarrow{(\sigma_1,\sigma_2)}(v_k,v_\ell)\right) = d_{H}(\sigma_1,\sigma_2)/s$.

And $\rateT(\delta)$ can be obtained by solving the following optimization problem \cite{Kolesnik1991, MarcusRoth1992}.
\vspace{-1mm}
\begin{equation}\label{eq:primal}
	\rateT(\delta)=\sup\left\{H(\va) : 
	\sum_{e\in \cE'}a_e=1,\, \sum_{e\in\cE'} a_eD(e)\le \delta  \right\}.
\end{equation}

To this end, we consider the dual problem of \eqref{eq:primal}. 
Specifically, we define a $(|\cV|^2\times |\cV|^2)$-{\em distance matrix} $\vT_{\cG\times \cG}(y)$ whose rows and columns are indexed by $\cV'$. 
For each entry indexed by $e\in \cV'\times \cV'$,
we set the entry to be zero if $e \notin \cE'$ and we set it to be $y^{D(e)}$ if $e\in \cE'$. 
Then the dual problem can be stated in terms of the dominant eigenvalue of the matrix $\vT_{\cG\times \cG}(y)$.

By applying the reduction techniques from~\cite{MarcusRoth1992}, we can reduce the problem size by a factor of two.
Formally, in the case of $s=1$, we define a $\binom{|\cV|+1}{2} \times \binom{|\cV|+1}{2}$- {\em reduced distance matrix} $\vB_{\cG\times \cG}(y)$ whose rows and columns are indexed by $\cV^{(2)} \triangleq \{(v_i,v_j): 1 \leq i \leq j \leq |\cV|\}$ using the following procedure.

Two states $s_1 = (v_i,v_j)$ and $s_2 = (v_k,v_\ell)$ in $\cG \times \cG$ are said to be {\em equivalent} if $v_i = v_\ell$ and $v_j = v_k$. 
The matrix $\vB_{\cG\times \cG}(y)$ is then obtained by merging all pairs of equivalent states $s_1$ and $s_2$. That is, we add the column indexed by $v_2$ to the column indexed by $v_1$ and then remove the row and column which are indexed by $v_2$. Note that it may be possible to reduce the size of this matrix $\vB_{\cG\times \cG}(y)$ further.
However, for the ease of exposition, we do not consider this case in this work.

Following this procedure, we observe that the entries in the matrix $\vB_{\cG\times \cG}(y)$ can be described by the rules in TABLE~\ref{table:BGG}. Moreover, the dominant eigenvalue of $\vB_{\cG\times \cG}(y)$ is the same as that of $\vT_{\cG\times \cG}(y)$. Then by strong duality, computing \eqref{eq:primal} is equivalent to solving the following dual problem~\cite{Luenberger.1973,Rockafellar.1970} (see also, \cite{Kashyap2019}). 
\vspace{-1mm}
\begin{equation}\label{eq:dual}
	\rateT(\delta)=\inf\left\{-\delta \log y +\log \Lambda(\vB_{\cG\times \cG}(y)) : 
	0\le y\le 1  \right\}.
\end{equation}

\noindent Here, we use $\Lambda(\vM)$ to denote the dominant eigenvalue of matrix $\vM$.
To simplify further, we write $\Lambda(y;\vB)\triangleq \Lambda(\vB_{\cG\times \cG}(y))$.

Since the objective function in \eqref{eq:dual} is convex,
it follows from standard calculus that any local minimum solution $y^*$ in the interval $[0,1]$ is also a global minimum solution.
Furthermore, $y^*$ is a zero of the first derivative of the objective function. If we consider the numerator of this derivative, then $y^*$ is a root of the function,
\begin{equation}\label{eq:derivative}
	F(y) \triangleq y\Lambda'(y;\vB) -\delta\Lambda(y;\vB) .
\end{equation}

In Corollary~\ref{cor:uniquesol}, we show that there is only one $y^*$ such that $F(y^*) = 0$ and $F'(y)$ is strictly positive for all values of $y$. Therefore, to evaluate the GV bound for a fixed $\delta$, it suffices to determine $y^*$.

Later, Marcus and Roth \cite{MarcusRoth1992} improve the GV bound \eqref{eq:GV} by considering certain subsets of the constrained space $\cS$. 
This entails the inclusion of an additional constraint defined in optimization problem \eqref{eq:primal}, and 
correspondingly, an additional variable in the dual problem \eqref{eq:dual}. 
Specifically, they considered certain subsets $\cS(p)\subseteq \cS$ 
where each symbol in the words of $\cS(p)$ appears with a certain frequency dependent on the parameter $p$. We describe this in more detail in Section~\ref{sec:mr}. 
\begin{table}[!t]
	\begin{center}
		\ra{1.5}
		\setlength{\tabcolsep}{4pt}
		\begin{tabular}{C{70mm} ccccc}
			\hline\hline
			{$\vB_{\cG\times\cG}(y)$ at entry $\left((v_i,\!v_j\!),(v_k,\!v_\ell\!)\right)$}
			&
			\multicolumn{5}{c}{Subgraph induced by the states $\{v_i,v_j,v_k,v_\ell\}$}
			\\ \hline \hline
			
			$0$ &
			\begin{tikzpicture}[scale=.7, transform shape]
				\tikzstyle{every node} = [draw,shape = circle]
				\node (i) at (0, 0) {$v_i$};
				\node (j) at (0, -1.5) {$v_j$};
				\node (k) at (1.5, 0) {$v_k$};
				\node (l) at (1.5, -1.5) {$v_\ell$};
			\end{tikzpicture} &
			\begin{tikzpicture}[scale=.7, transform shape]
				\tikzstyle{every node} = [draw,shape = circle]
				\node (i) at (0, 0) {$v_i$};
				\node (j) at (0, -1.5) {$v_j$};
				\node (k) at (1.5, 0) {$v_k$};
				\node (l) at (1.5, -1.5) {$v_\ell$};
				\draw [->] (i) -- (k);
			\end{tikzpicture} &
			\begin{tikzpicture}[scale=.7, transform shape]
				\tikzstyle{every node} = [draw,shape = circle]
				\node (i) at (0, 0) {$v_i$};
				\node (j) at (0, -1.5) {$v_j$};
				\node (k) at (1.5, 0) {$v_k$};
				\node (l) at (1.5, -1.5) {$v_\ell$};
				\draw [->] (i) -- (l);
			\end{tikzpicture} 
			&
			\begin{tikzpicture}[scale=.7, transform shape]
				\tikzstyle{every node} = [draw,shape = circle]
				\node (i) at (0, 0) {$v_i$};
				\node (j) at (0, -1.5) {$v_j$};
				\node (k) at (1.5, 0) {$v_k$};
				\node (l) at (1.5, -1.5) {$v_\ell$};
				\draw [->] (j) -- (k);
			\end{tikzpicture} &
			\begin{tikzpicture}[scale=.7, transform shape]
				\tikzstyle{every node} = [draw,shape = circle]
				\node (i) at (0, 0) {$v_i$};
				\node (j) at (0, -1.5) {$v_j$};
				\node (k) at (1.5, 0) {$v_k$};
				\node (l) at (1.5, -1.5) {$v_\ell$};
				\draw [->] (j) -- (l);
			\end{tikzpicture} \\ \hline
			1 & 
			\begin{tikzpicture}[scale=.7, transform shape]
				\tikzstyle{every node} = [draw,shape = circle]
				\node (i) at (0, 0) {$v_i$};
				\node (j) at (0, -1.5) {$v_j$};
				\node (k) at (1.5, 0) {$v_k$};
				\node (l) at (1.5, -1.5) {$v_\ell$};
				\draw [->] (i) -- (k) node[pos = 0.5,above,draw = none] {$\sigma$};
				\draw [->] (j) -- (l) node[pos = 0.5,above,draw = none] {$\sigma$};
			\end{tikzpicture} &
			\begin{tikzpicture}[scale=.7, transform shape]
				\tikzstyle{every node} = [draw,shape = circle]
				\node (i) at (0, 0) {$v_i$};
				\node (j) at (0, -1.5) {$v_j$};
				\node (k) at (1.5, 0) {$v_k$};
				\node (l) at (1.5, -1.5) {$v_\ell$};
				\draw [->] (i) -- (l) node[pos = 0.3,above,draw = none] {$\sigma$};
				\draw [->] (j) -- (k) node[pos = 0.1,above,draw = none] {$\sigma$};
			\end{tikzpicture} &
			\begin{tikzpicture}[scale=.7, transform shape]
				\tikzstyle{every node} = [draw,shape = circle]
				\node (i) at (0, 0) {$v_i$};
				\node (j) at (0, -1.5) {$v_j$};
				\node (k) at (1.5, -0.75) {$v_k$};
				\draw [->] (i) -- (k) node[pos = 0.5,above,draw = none] {$\sigma$};
				\draw [->] (j) -- (k) node[pos = 0.5,above,draw = none] {$\sigma$};
			\end{tikzpicture}  \\ \hline
			$y$ &
			\begin{tikzpicture}[scale=.7, transform shape]
				\tikzstyle{every node} = [draw,shape = circle]
				\node (i) at (0, 0) {$v_i$};
				\node (j) at (0, -1.5) {$v_j$};
				\node (k) at (1.5, 0) {$v_k$};
				\node (l) at (1.5, -1.5) {$v_\ell$};
				\draw [->] (i) -- (k) node[pos = 0.5,above,draw = none] {$\sigma$};
				\draw [->] (j) -- (l) node[pos = 0.5,above,draw = none] {$\bar{\sigma}$};
			\end{tikzpicture} &
			\begin{tikzpicture}[scale=.7, transform shape]
				\tikzstyle{every node} = [draw,shape = circle]
				\node (i) at (0, 0) {$v_i$};
				\node (j) at (0, -1.5) {$v_j$};
				\node (k) at (1.5, 0) {$v_k$};
				\node (l) at (1.5, -1.5) {$v_\ell$};
				\draw [->] (i) -- (l) node[pos = 0.3,above,draw = none] {$\sigma$};
				\draw [->] (j) -- (k) node[pos = 0.1,above,draw = none] {$\bar{\sigma}$};
			\end{tikzpicture} &
			\begin{tikzpicture}[scale=.7, transform shape]
				\tikzstyle{every node} = [draw,shape = circle]
				\node (i) at (0, 0) {$v_i$};
				\node (j) at (0, -1.5) {$v_j$};
				\node (k) at (1.5, -0.75) {$v_k$};
				\draw [->] (i) -- (k) node[pos = 0.5,above,draw = none] {$\sigma$};
				\draw [->] (j) -- (k) node[pos = 0.5,above,draw = none] {$\bar{\sigma}$};
			\end{tikzpicture}\\ \hline
			$2y$ &
			\begin{tikzpicture}[scale=.7, transform shape]
				\tikzstyle{every node} = [draw,shape = circle]
				\node (k) at (0, 0) {$v_k$};
				\node (l) at (0, -1.5) {$v_\ell$};
				\node (i) at (-1.5, -0.75) {$v_i$};
				\draw [->] (i) -- (k) node[pos = 0.5,above,draw = none] {$\sigma$};
				\draw [->] (i) -- (l) node[pos = 0.5,above,draw = none] {$\bar{\sigma}$};
			\end{tikzpicture}  
			\\ \hline
		\end{tabular}
		\vspace{-2mm}
	\end{center}
	\caption{We set the $\left((v_i,v_j),(v_k,v_\ell)\right)$-entry of the matrix $\vB_{\cG\times \cG}(y)$ according to subgraph induced by the states $v_i$,$v_j$,$v_k$, and $v_\ell$. Here, $\bar{\sigma}$ denotes the complement of $\sigma$.\vspace{-5mm}  }
	\label{table:BGG}
\end{table}

\subsection{Our Contributions}

\begin{enumerate}[(A)]
	\item In Section~\ref{sec:gv}, we develop the numerical procedures to compute $\rateT(\delta)$ for a fixed $\delta$ and hence, determine the GV bound \eqref{eq:GV}. Our procedure modifies the well-known {\em power iteration method} to compute the derivatives of $\Lambda(y;\vB)$. After that, using these derivatives, we apply the classical Newton-Raphson method to determine the root of \eqref{eq:derivative}.
	In the same section, we also study procedures to plot the GV curve, that is, the set  $\{(\delta,\rgv(\delta)): 0\le \delta\le 1\}$.
	Here, we demonstrate that the GV curve can be plotted without any Newton-Raphson iteration.	
	
	
	\item In Section~\ref{sec:mr}, we then develop similar power iteration methods and numerical procedures to compute the GV-MR bound. Similar to the GV curve, we also provide a plotting procedure that uses significantly less Newton-Raphson iterations.
	
	\item In Section~\ref{sec:ss}, we provide explicit formulas for the computation of GV bound and GV-MR bound for graph presentations that have exactly one state but multiple parallel edges. 
	
	\item In Section~\ref{sec:numerical}, we validate our methods by computing the GV and the GV-MR bounds for some specific constrained systems. 
	For comparison purposes, we also plot a simple lower bound that is obtained by using an upper estimate of the ball size. 
	From the plots in Figures~\ref{fig:swcc},~\ref{fig:rll} and~\ref{fig:secc}, it is also clear that the GV and GV-MR bounds are significantly better.
	We also observe that the GV bound and GV-MR bound for {\em subblock energy constrained codes (SECC)} obtained through our procedures improve the GV-type bound given by Tandon et al.~\cite[Proposition~12]{Tandon2018}. 
	
\end{enumerate}

\section{Evaluating the Gilbert-Varshamov Bound}
\label{sec:gv}

In this section, we first describe a numerical procedure that solves \eqref{eq:dual} and hence determine $\rgv(\delta)$ for fixed values of $\delta$. Then we showed the procedure can be simplified when we compute the GV-curve, that is, the set of points $\{(\delta,\rgv(\delta)) : \delta\in \interval{0,1}\}$. 
Here, we abuse notation and use $\interval{a,b}$ to denote the interval $\{x : a\le x\le b\}$, if $a<b$; and the interval $\{x : b\le x\le a\}$, otherwise.

Below we provide formal description of our procedure to obtain the GV bound for a fixed relative distance $\delta$.

\vspace{2mm}

\noindent{\bf Procedure 1 (GV bound for fixed relative distance)}.

\noindent{\sc Input}: 
Adjacency matrix $\vA_{\cG}$, reduced distance matrix $\vB_{\cG\times\cG}(y)$, and relative minimum distance $\delta$\\[1mm]  
\noindent{\sc Output}: GV bound, that is, $\rgv(\delta)$ as defined in \eqref{eq:GV}  
\begin{enumerate}[(1)]
	\item Apply the Newton-Raphson method to obtain $y^*$ such that $F(y^*)$ is approximately zero. 
	\begin{itemize}
		\item Fix some tolerance value $\epsilon$.
		\item Set $t = 0$ and pick some initial guess $0\le y_t\le 1$.
		\item While $|y_{t}-y_{t-1}|>\epsilon$~:
		\begin{itemize}
			\item Compute the next guess $y_{t+1}$ as follows.
			\begin{equation*}
				\hspace*{-10mm}
				y_{t+1} =y_t-\dfrac{F(y_t)}{F'(y_t)}
				= y_t-\dfrac{y_t\Lambda'(y_t;\vB)-\delta\Lambda(y_t;\vB))}{(1-\delta)\Lambda'(y_t;\vB)+y_t\Lambda''(y_t;\vB)}\,.
			\end{equation*}
			\item In this step, we apply the power iteration method to compute $\Lambda(y_t;\vB)$, $\Lambda'(y_t;\vB)$, and $\Lambda''(y_t;\vB)$.
			\item Increment $t$ by one.
		\end{itemize}
		\item Set $y^*\gets y_t$.
	\end{itemize}
	\item Determine $\rgv(\delta)$ using $y^*$. Specifically, we compute:
	$\rateT(\delta) \triangleq -\delta\log y^* + \log \Lambda(y^*;\vB)$,
	$\capy(\cS) \triangleq \log \Lambda(\vA_\cG)$, and
	$\rgv(\delta) \triangleq 2\capy(\cS) - \rateT(\delta)$.
\end{enumerate}

Throughout Sections~\ref{sec:gv} and~\ref{sec:mr}, we illustrate our numerical procedures via a running example using the class of {\em sliding window constrained codes (SWCC)}. 
Formally, we fix a {\em window length} $L$ and {\em window weight} $w$, and say that a binary word satisfies the $(L,w)$-{\em sliding window weight-constraint} if the number of ones in every consecutive $L$ bits is at least $w$. We refer to the collection of words that meet this constraint as an $(L,w)$-{\em SWCC constrained system}.
The class of SWCCs was introduced by Tandon \etal{} for the
application of simultaneous energy and information transfer \cite{Tandon2016, Wu2021}. 
Later, Immink and Cai \cite{Immink2020,Immink.2020} studied encoders for this constrained system and provided a simple graph presentation that uses only $\binom{L}{w}$ states.

In the next example, we illustrate how the numerical procedure can be used to compute the GV bound for the value when $\delta=0.1$.

\begin{example}
	Let $L=3$ and $w=2$ and we consider $(3,2)$-SWCC constrained system.
	From \cite{Immink2020}, we have the following graph presentation with states {\tt x11}, {\tt 101}, and {\tt 110}.
	
	\begin{center}
		\begin{tikzpicture}[scale=.8]
			\tikzstyle{every node} = [draw,shape = circle]
			\node (x11) at (0, 0) {{\tt x11}};
			\node (110) at (3, 0) {{\tt 110}};
			\node (101) at (6, 0) {{\tt 101}};
			\draw [->] (x11) to node[pos = 0.5,below,draw = none] {{\tt 0}} (110);
			\draw [->] (110) to node[pos = 0.5,below,draw = none] {{\tt 1}} (101);
			\draw [->] (101) to [bend right=40]  node[pos =0.5, above, draw = none] {{\tt 1}} (x11);
			\draw [->] (x11) to [out=135, in=180, loop]  node[pos =0.25, above, draw = none] {{\tt 1}} (x11);
		\end{tikzpicture}
	\end{center}
	
	Then the corresponding adjacency and reduced distance matrices are as follows:
	\[
	\vA_{\cG}  = 
	\begin{bmatrix}
		1 & 1 & 0\\
		0 & 0 & 1\\
		1 & 0 & 0
	\end{bmatrix},\,
	\vB_{\cG\times \cG} (y) = 
	\begin{bmatrix}
		1 & 2y & 0 & 1 & 0 & 0\\
		0 & 0 & 1 & 0 & y & 0\\
		1 & y & 0 & 0 & 0 & 0\\
		0 & 0 & 0 & 0 & 0 & 1\\
		0 & 0 & 1 & 0 & 0 & 0\\
		1 & 0 & 0 & 0 & 0 & 0 \\
	\end{bmatrix}\,.
	\]
	To determine the GV bound at $\delta = 0.1$, we first approximate the optimal point $y^*$ for which $-\delta \log y + \log \Lambda(y;\vB)$ is minimized.
	
	We apply the Newton-Raphson method to find a zero of the function $F(y)$.
	Now, with the initial guess $y_0=0.3$, we apply the power iteration method to determine 
	\[\Lambda(0.3;\vB)=1.659,~\Lambda'(0.3;\vB)=0.694,~\Lambda''(0.3;\vB)=0.183.\]
	\noindent Then we compute that $y_1 \approx 0.238$. Repeating the computations, we have that $y_2\approx 0
	238$. Since $|y_2-y_1|$ is less than the tolerance value $10^{-5}$, we set $y^*=0.238$.
	Hence, we have that $\rateT(0.1)= 0.9$. 
	Applying the power iteration method to either $\vA_\cG$ or $\vB_{\cG\times \cG}(0)$, we compute the capacity of $(3,2)$-SWCC constrained system to be $\capy(\cS) = 0.551$.
	Then the GV bound is given by $\rgv(0.1)=2(0.551)-0.9=0.202$. \hfill\IEEEQEDhere
\end{example}

We discuss the convergence issues arising from Procedure 1.
Observe that there are two different iterative processes in Step 1: namely,
(a)~the power iteration method to compute the values  $\Lambda(y_t;\vB)$, $\Lambda'(y_t;\vB)$, and $\Lambda''(y_t;\vB)$; and 
(b)~the Newton-Raphson method that determines the zero of $F(y)$.

\begin{enumerate}[(a)]
	\item Recall that $\Lambda(y;\vB)$ is the largest eigenvalue of the reduced distance matrix $\vB_{\cG\times \cG}(y)$. 
	If we apply naive methods to compute this dominant eigenvalue,  the computational complexity increases very rapidly with the matrix size. 
	Specifically, if $\cG$ has $M$ states, then the reduced distance matrix has dimension $\Theta(M^2)\times \Theta(M^2)$ and finding its characteristic equation takes $O(M^6)$ time. Even then, determining the exact roots of characteristic equations with degree at least five is generally impossible. Therefore, we turn to the numerical procedures like the ubiquitous power iteration method~\cite{Stewart.1973}.
	
	However, the standard power iteration method is only able to compute the dominant eigenvalue $\Lambda(y;\vB)$. 
	Nevertheless, we can modify the power iteration method to compute $\Lambda(y;\vB)$ and its higher order derivatives.
	In Appendix~\ref{app:power}, we demonstrate that under certain mild assumptions, the modified power iteration method always converges.
	Moreover, using the sparsity of the reduced distance matrix, we have that each iteration can be completed in $O(M^2)$ time.
	\item Next, we discuss whether we can guarantee that $y_t$ converges to $y^*$ as $t$ approaches infinity. 
	Even though the Newton-Raphson method converges in all our numerical experiments, we are unable to demonstrate that it always converges for $F(y)$.
	Nevertheless, we can circumvent this issue if we are interested in {\em plotting the GV curve}. Specifically, if our objective is to determine the curve $\{ (\delta,\rgv(\delta)): \delta\in\interval{0,1} \}$, it turns out that we need not implement the Newton-Raphson iterations and we discuss this next.
\end{enumerate}


Fix some constrained system $\cS$. Let us define its corresponding {\em GV curve} to be the set of points
$\gvcurve(\cS)\triangleq \{(\delta,\rgv(\delta)): \delta\in\interval{0,1} \}$.
Here, we demonstrate that the GV curve can be plotted without any Newton-Raphson iteration.

To this end, we observe that when $F(y^*)=0$, we have that $\delta = y^*\Lambda'(y^*;\vB)/\Lambda(y^*;\vB)$. Hence, we abuse notation and define the function 
\begin{equation}\label{eq:delta}
	\delta(y) \triangleq y\Lambda'(y;\vB)/\Lambda(y;\vB)\,.
\end{equation}
We further define $\delta_{\rm max} = \delta(1)=\Lambda'(1;\vB)/\Lambda(1;\vB)$.
In this section, we prove the following theorem.
\begin{theorem}\label{thm:curve}
	Let $\cG$ be the graph presentation for the constrained system $\cS$.
	If we define the function
	\begin{equation}\label{eq:rho}
		\rho_{\rm GV}(y) \triangleq 2\capy(\cS) + \delta(y)\log y - \log \Lambda(y;\vB)\,,
	\end{equation}
	then the corresponding GV curve is given by 
	\begin{equation} \label{eq:gvcurve}
		\gvcurve(\cS) = \big\{(\delta(y), \rho_{\rm GV}(y)) : y\in [0,1] \big\} 
		\cup \big\{ (\delta, 0) : \delta\ge \delta_{\rm max} \big\}\,. 
	\end{equation}
\end{theorem}

Before we prove Theorem~\ref{thm:curve}, we discuss its implications.
Note that to compute $\delta(y)$ and $\rho(y)$, it suffices to determine $\Lambda(y;\vB)$ and $\Lambda'(y;\vB)$ using the modified power iteration methods described in Appendix~\ref{app:power}.
In other words, no Newton-Raphson iterations are required. 
We also have additional computational savings as we need not apply the power iteration method to compute the second derivative $\Lambda''(y;\vB)$.

\begin{example}
	We continue our example and plot the GV curve for the $(3,2)$-SWCC constrained system in Fig.~\ref{fig:swcc}(a). 
	Before plotting, we observe that when $y=0$, we have $(\delta(0),\rho(0))=(0,0.551)=(0,\capy(\cS))$, as expected.
	When $y=1$, we have $\delta(1)=\delta_{\rm max} = 0.313$.
	Indeed, both $\rho(1)$ and $\rgv(\delta_{\rm max})$ are equal to zero and we have that $\rgv(\delta) = 0$ for $\delta\ge \delta_{\rm max}$.
	
	Next, we compute a set of 100 points on the GV curve.
	If we apply Procedure 1 to compute $\rgv(\delta)$ for 100 values of $\delta$ in the interval $[0,\delta_{\rm max}]$,
	we require 275 Newton-Raphson iterations and 6900 power iterations to find these points.
	In contrast, applying Theorem~\ref{thm:curve}, we compute $(\delta(y),\rho(y))$ for 100 values of $y$ in the interval $\interval{0,1}$. 
	This does not require any Newton-Raphson iterations and involves only 2530 power iterations.
	\hfill \IEEEQEDhere
\end{example}

	To prove Theorem~\ref{thm:curve}, we demonstrate the following lemmas.
	Our first lemma is immediate from the definitions of $\rgv$, $\delta$ and $\rho$ 
	in \eqref{eq:GV}, \eqref{eq:delta}, and \eqref{eq:rho}, respectively.
	
	\begin{lemma}\label{lem:rgv-rho}
		$\rgv(\delta(y))=\rho(y)$ for all $y\in[0,1]$.
	\end{lemma}
	
	The next lemma studies the behaviour of both $\delta$ and $\rho$ as functions in $y$.
	
	\begin{lemma}\label{lem:delta-rho}
		In terms of $y$, the functions $\delta(y)$ and $\rho(y)$ are monotone increasing and decreasing, respectively.
		Furthermore, we have that 
		$(\delta(0),\rho(0))=(0,\capy(\cS))$, $(\delta(1),\rho(1))=(\delta_{\rm max},0)$ and $\rgv(\delta)=0$ for $\delta\ge \delta_{\rm max}$. 
	\end{lemma}
	
	\begin{IEEEproof}To simplify notation, we write $\Lambda(y;\vB)$, $\Lambda'(y;\vB)$ and $\Lambda''(y;\vB)$
		as $\Lambda$, $\Lambda'$, and $\Lambda''$, respectively.
		
		First, we show that $\delta'(y)$ is positive for $0\le y<1$.
		Differentiating the expression in \eqref{eq:delta}, we have that $\delta'(y)>0$ is equivalent to
		\begin{equation}\label{eq:delta-der}
			\Lambda(\Lambda'+y\Lambda'')-y(\Lambda')^2 >0.
		\end{equation}

		Recall that \eqref{eq:dual} is a convex minimization problem.
		Hence, the second order derivative of the objective function is always positive.
		In other words, 
		\begin{equation*}
			\frac{\delta}{y^2}+\frac{\Lambda''\Lambda-(\Lambda')^2}{\Lambda^2}>0.
		\end{equation*}
		Substituting $\delta$ with $y\Lambda'/\Lambda$ and multiplying by $y\Lambda^2$, 
		we obtain \eqref{eq:delta-der} as desired.
		
		Next, we show that $\rho$ is monotone decreasing. 
		Recall that $\rho(y)=\rgv(\delta(y))=\capy(\cS)-\rateT(\delta)$.
		Since $\rateT(\delta)$ yields the asymptotic rate of the total ball size,
		we have that as $y$ increases, $\delta(y)$ increases and so,  $\rateT(\delta)$ increases.
		Therefore, $\rho(y)$  decreases, as desired.
		
		Next, we show that $\rho(1)=0$. 
		When $y=1$, we have from \eqref{eq:rho} that $\rho(1)=2\capy(\cS) - \log \Lambda(1;\vB)$.
		Now, recall that $\vB_{\cG\times\cG}(y)$ shares the same dominant eigenvalue as the matrix $\vT_{\cG\times\cG}(y)$ \cite{Kolesnik1991}. 
		Furthermore it can be verified that, when $y=1$, $\vT_{\cG\times\cG}(1)$ is tensor product of $\vA_\cG$ and $\vA_\cG$. 
		That is, $\vT_{\cG\times\cG}(1)=\vA_\cG\otimes \vA_\cG$.
		It then follows from standard linear algebra that $\Lambda(1;\vB)=\Lambda(1;\vT)=\Lambda(\vA_\cG)^2$. 
		Thus, $\log \Lambda(1;\vB) = 2\capy(\cS)$ and $\rho(1)=0$.
		In this instance, we also have that $\rateT(\delta_{\rm max})=2\capy(\cS)$.
		
		Finally, for $\delta\ge\delta_{\rm max}$, we have that $\rateT(\delta_{\rm max})=2\capy(\cS)$ and thus, $\rgv(\delta)=0$, as required.
	\end{IEEEproof}
	
	Theorem~\ref{thm:curve} is then immediate from Lemmas~\ref{lem:rgv-rho} and~\ref{lem:delta-rho}. 
	
	We have the following corollary that immediately follows from Lemma~\ref{lem:delta-rho}. This corollary then implies that $y^*$ yields the global minimum for the optimization problem.
	
	\begin{corollary}\label{cor:uniquesol}
		When $0 \le \delta \le \delta_{max} = \frac{\Lambda'(1,\vB)}{\Lambda(1,\vB)}$, $F(y) \triangleq y\Lambda'(y;\vB) -\delta\Lambda(y;\vB)$ has a unique zero in $\interval{0,1}$. Furthermore, $F'(y)$ is strictly positive for all $y \in \interval{0,1}$.
	\end{corollary}

\section{Evaluating Marcus and Roth's Improvement of the the Gilbert-Varshamov Bound}
\label{sec:mr}

In \cite{MarcusRoth1992}, Marcus and Roth improve the GV lower bound for most constrained systems by considering subsets $\cS(p)$ of $\cS$ where $p$ is some parameter.
Here, we focus on the case $s=1$ and set $p$ to be the normalized frequency of edges whose labels correspond to one.
Specifically, we set $\cS(p) \triangleq \{\vx\in \cS: {\rm wt}(\vx) = \floor{p|\vx|}\}$. 

Next, let $\cS_{n}(p)$ to be the set of all words/paths of length $n$ in $\cS(p)$ and
we define $S(p) \triangleq \limsup_{n\to\infty} \frac1n\log |\cS_n(p)|$.

Similar to before, we define $\rateT(p,\delta) = \limsup_{n\to\infty} \frac1n{\log T(\floor{\delta n},n;\cS_n(p))}$.
Since $\cS_n(p)$ is a subset of $\cS_n$, it follows the usual GV argument that there exists a family of $(n,\lfloor \delta n \rfloor; \cS)$-codes whose rates approach $2S(p)-\rateT(p,\delta)$ for all $0\le p\le 1$. 
Therefore, we have the following lower bound on asymptotic achievable code rates: 
\begin{equation}\label{eq:MR}
	\rmr(\delta) = \sup \{2S(p)-\rateT(p,\delta) : 0\le p\le 1\}\,.
\end{equation}	

Now, a key result from \cite{MarcusRoth1992} is that both $S(p)$ and $\rateT(p,\delta)$ can be obtained via two different convex optimization problems. For succinctness, we state the {\em dual} formulations of these optimization problems.

First, $S(p)$ can be obtained from the following problem.
\begin{equation}\label{eq:Sp}
	S(p)=\inf\left\{-p \log z + \log \Lambda(\vC_{\cG}(z)) : 
	z\ge 0  \right\}.
\end{equation}				

Here, $\vC_\cG(z)$ is the following $(|\cV| \times |\cV|)$-matrix $\vC_{\cG}(z)$ whose rows and columns are indexed by $\cV$. 
For each entry indexed by $e$, we set $(\vC_{\cG}(z))_e$  to be zero if $e \notin \cE$, and $z^{\cL(e)}$ if $e \in \cE$.

As before, we simplify notation by writing $\Lambda(z;\vC)\triangleq \Lambda(\vC_{\cG}(z))$. 
Again, following the convexity of \eqref{eq:Sp},
we are interested in finding the zero of the following function.

\begin{equation}\label{eq:constraint1}
	G_1(z) \triangleq z\Lambda'(z;\vC)-p\Lambda(z;\vC).
\end{equation}

Next, $\rateT(p,\delta)$ can be obtained via the following optimization.
\begin{equation}
\rateT(p,\delta)  = \inf\big\{-2p \log x -\delta \log y +\log \Lambda(\vD_{\cG\times \cG}(x,y)) : ~x\ge 0,~0\le y\le 1  \big\}. \label{eq:Tp}
\end{equation}
Here, $\vD_{\cG\times \cG}(x,y)$ is a $\binom{|\cV|+1}{2} \times \binom{|\cV|+1}{2}$-reduced distance matrix indexed by $\cV^{(2)}$. 
To define the entry of matrix $\vD_{\cG\times \cG}(x,y)$ indexed by $((v_i,v_j),(v_k,v_\ell))$, we look at the vertices $v_i$, $v_j$, $v_k$ and $v_\ell$ and follow the rules given in Table~\ref{table:DGG}.

\begin{table*}[t]
	\begin{center}
		\ra{1.5}
		\begin{tabular}{c ccccc}
			\hline\hline
			$\vD_{\cG\times \cG}(x,y)$ at entry $\left((v_i,v_j),(v_k,v_\ell)\right)$
			&
			\multicolumn{5}{c}{Subgraph induced by the states $\{v_i,v_j,v_k,v_\ell\}$}
			\\ \hline \hline
			$0$ &
			\begin{tikzpicture}[scale=.7, transform shape]
				\tikzstyle{every node} = [draw,shape = circle]
				\node (i) at (0, 0) {$v_i$};
				\node (j) at (0, -1.5) {$v_j$};
				\node (k) at (1.5, 0) {$v_k$};
				\node (l) at (1.5, -1.5) {$v_\ell$};
			\end{tikzpicture} &
			\begin{tikzpicture}[scale=.7, transform shape]
				\tikzstyle{every node} = [draw,shape = circle]
				\node (i) at (0, 0) {$v_i$};
				\node (j) at (0, -1.5) {$v_j$};
				\node (k) at (1.5, 0) {$v_k$};
				\node (l) at (1.5, -1.5) {$v_\ell$};
				\draw [->] (i) -- (k);
			\end{tikzpicture} &
			\begin{tikzpicture}[scale=.7, transform shape]
				\tikzstyle{every node} = [draw,shape = circle]
				\node (i) at (0, 0) {$v_i$};
				\node (j) at (0, -1.5) {$v_j$};
				\node (k) at (1.5, 0) {$v_k$};
				\node (l) at (1.5, -1.5) {$v_\ell$};
				\draw [->] (i) -- (l);
			\end{tikzpicture} &
			\begin{tikzpicture}[scale=.7, transform shape]
				\tikzstyle{every node} = [draw,shape = circle]
				\node (i) at (0, 0) {$v_i$};
				\node (j) at (0, -1.5) {$v_j$};
				\node (k) at (1.5, 0) {$v_k$};
				\node (l) at (1.5, -1.5) {$v_\ell$};
				\draw [->] (j) -- (k);
			\end{tikzpicture} &
			\begin{tikzpicture}[scale=.7, transform shape]
				\tikzstyle{every node} = [draw,shape = circle]
				\node (i) at (0, 0) {$v_i$};
				\node (j) at (0, -1.5) {$v_j$};
				\node (k) at (1.5, 0) {$v_k$};
				\node (l) at (1.5, -1.5) {$v_\ell$};
				\draw [->] (j) -- (l);
			\end{tikzpicture} \\ \hline
			$1$ &
			\begin{tikzpicture}[scale=.7, transform shape]
				\tikzstyle{every node} = [draw,shape = circle]
				\node (i) at (0, 0) {$v_i$};
				\node (j) at (0, -1.5) {$v_j$};
				\node (k) at (1.5, 0) {$v_k$};
				\node (l) at (1.5, -1.5) {$v_\ell$};
				\draw [->] (i) -- (k) node[pos = 0.5,above,draw = none] {$0$};
				\draw [->] (j) -- (l) node[pos = 0.5,above,draw = none] {$0$};
			\end{tikzpicture} &
			\begin{tikzpicture}[scale=.7, transform shape]
				\tikzstyle{every node} = [draw,shape = circle]
				\node (i) at (0, 0) {$v_i$};
				\node (j) at (0, -1.5) {$v_j$};
				\node (k) at (1.5, 0) {$v_k$};
				\node (l) at (1.5, -1.5) {$v_\ell$};
				\draw [->] (i) -- (l) node[pos = 0.3,above,draw = none] {$0$};
				\draw [->] (j) -- (k) node[pos = 0.1,above,draw = none] {$0$};
			\end{tikzpicture} &
			\begin{tikzpicture}[scale=.7, transform shape]
				\tikzstyle{every node} = [draw,shape = circle]
				\node (i) at (0, 0) {$v_i$};
				\node (j) at (0, -1.5) {$v_j$};
				\node (k) at (1.5, -0.75) {$v_k$};
				\draw [->] (i) -- (k) node[pos = 0.5,above,draw = none] {$0$};
				\draw [->] (j) -- (k) node[pos = 0.5,above,draw = none] {$0$};
			\end{tikzpicture} \\ \hline
			$x^2$ &
			\begin{tikzpicture}[scale=.7, transform shape]
				\tikzstyle{every node} = [draw,shape = circle]
				\node (i) at (0, 0) {$v_i$};
				\node (j) at (0, -1.5) {$v_j$};
				\node (k) at (1.5, 0) {$v_k$};
				\node (l) at (1.5, -1.5) {$v_\ell$};
				\draw [->] (i) -- (k) node[pos = 0.5,above,draw = none] {$1$};
				\draw [->] (j) -- (l) node[pos = 0.5,above,draw = none] {$1$};
			\end{tikzpicture} &
			\begin{tikzpicture}[scale=.7, transform shape]
				\tikzstyle{every node} = [draw,shape = circle]
				\node (i) at (0, 0) {$v_i$};
				\node (j) at (0, -1.5) {$v_j$};
				\node (k) at (1.5, 0) {$v_k$};
				\node (l) at (1.5, -1.5) {$v_\ell$};
				\draw [->] (i) -- (l) node[pos = 0.3,above,draw = none] {$1$};
				\draw [->] (j) -- (k) node[pos = 0.1,above,draw = none] {$1$};
			\end{tikzpicture} &
			\begin{tikzpicture}[scale=.7, transform shape]
				\tikzstyle{every node} = [draw,shape = circle]
				\node (i) at (0, 0) {$v_i$};
				\node (j) at (0, -1.5) {$v_j$};
				\node (k) at (1.5, -0.75) {$v_k$};
				\draw [->] (i) -- (k) node[pos = 0.5,above,draw = none] {$1$};
				\draw [->] (j) -- (k) node[pos = 0.5,above,draw = none] {$1$};
			\end{tikzpicture} \\ \hline
			$xy$ &
			\begin{tikzpicture}[scale=.7, transform shape]
				\tikzstyle{every node} = [draw,shape = circle]
				\node (i) at (0, 0) {$v_i$};
				\node (j) at (0, -1.5) {$v_j$};
				\node (k) at (1.5, 0) {$v_k$};
				\node (l) at (1.5, -1.5) {$v_\ell$};
				\draw [->] (i) -- (k) node[pos = 0.5,above,draw = none] {$\sigma$};
				\draw [->] (j) -- (l) node[pos = 0.5,above,draw = none] {$\bar{\sigma}$};
			\end{tikzpicture} &
			\begin{tikzpicture}[scale=.7, transform shape]
				\tikzstyle{every node} = [draw,shape = circle]
				\node (i) at (0, 0) {$v_i$};
				\node (j) at (0, -1.5) {$v_j$};
				\node (k) at (1.5, 0) {$v_k$};
				\node (l) at (1.5, -1.5) {$v_\ell$};
				\draw [->] (i) -- (l) node[pos = 0.3,above,draw = none] {$\sigma$};
				\draw [->] (j) -- (k) node[pos = 0.1,above,draw = none] {$\bar{\sigma}$};
			\end{tikzpicture}&
			\begin{tikzpicture}[scale=.7, transform shape]
				\tikzstyle{every node} = [draw,shape = circle]
				\node (i) at (0, 0) {$v_i$};
				\node (j) at (0, -1.5) {$v_j$};
				\node (k) at (1.5, -0.75) {$v_k$};
				\draw [->] (i) -- (k) node[pos = 0.5,above,draw = none] {$\sigma$};
				\draw [->] (j) -- (k) node[pos = 0.5,above,draw = none] {$\bar{\sigma}$};
			\end{tikzpicture}\\ \hline
			$2xy$ &
			\begin{tikzpicture}[scale=.7, transform shape]
				\tikzstyle{every node} = [draw,shape = circle]
				\node (k) at (0, 0) {$v_k$};
				\node (l) at (0, -1.5) {$v_\ell$};
				\node (i) at (-1.5, -0.75) {$v_i$};
				\draw [->] (i) -- (k) node[pos = 0.5,above,draw = none] {$\sigma$};
				\draw [->] (i) -- (l) node[pos = 0.5,above,draw = none] {$\bar{\sigma}$};
			\end{tikzpicture}  \\ \hline
		\end{tabular}
	\end{center}
	\caption{We set the $\left((v_i,v_j),(v_k,v_\ell)\right)$-entry of the matrix $\vD_{\cG\times \cG}(x,y)$ according to subgraph induced by the states $v_i$,$v_j$,$v_k$, and $v_\ell$. }
	\label{table:DGG}
\end{table*}

Again, we write $\Lambda(x,y;\vD)\triangleq \Lambda(\vD_{\cG\times \cG}(x,y))$. 
Also, following the convexity of \eqref{eq:Tp}, 
we have that if the optimal solution is obtained at $x$ and $y$,
then 
\begin{align}
	G_2(x,y) & \triangleq x\Lambda_{x}(x,y;\vD)-2p\Lambda(x,y;\vD)=0. \label{constraint2}\\
	G_3(x,y) & \triangleq y\Lambda_{y}(x,y;\vD)-\delta\Lambda(x,y;\vD)=0.\label{constraint3}
\end{align}

To this end, we consider the function $\Delta(x) = \Lambda_y(x,1;\vD)/\Lambda(x,1;\vD)$ for $x>0$ and set $\delta_{\max} = \sup\{\Delta(x): x>0\}$. As with the previous section, we develop a numerical procedure to solve the optimization problem \eqref{eq:MR}.
To this end, we have the following critical observation. 

\begin{theorem}\label{thm:lagrange}
	For given $\delta < \delta_{\rm max}$, consider the optimization problem:
	\begin{align*}
		&	\sup\Big\{-2p\log z + 2\log \Lambda(z;\vC) + 2p\log x + \delta \log y - \log \Lambda(x,y;\vD) :\\
		&\hspace{90mm} G_1(z) = G_2(x,y) = G_3(x,y) = 0\Big\}\,.
	\end{align*}
	If $(p^*,x^*,y^*,z^*)$ is an optimal solution, then $x^* = z^*$. Furthermore, if $0 \le p^* \le 1$, then $x^*,z^* \ge 0$ and $0 \le y^* \le 1$. 
\end{theorem}

\begin{proof}
	Let $\lambda_1, \lambda_2, \lambda_3$ be real-valued variables and 
	define $L(p,x,y,z,\lambda_1, \lambda_2, \lambda_3) \triangleq G(p,x,y,z) + \lambda_1G_1(z) + \lambda_2G_2(x,y) + \lambda_3G_3(x,y)$. 
	Using  Lagrangian multiplier theorem, we have that $\partial L/\partial p = \partial L/\partial x = \partial L/\partial y = \partial L/\partial z = 0$ for any optimal solution.
	Solving these equations with the constraints $G_1(z) = G_2(x,y) = G_3(x,y) = 0$, we have that $\lambda_1 = \lambda_2 = \lambda_3 = 0$ and $x=z$ for any optimal solution.
	
	Now, when $p^*\in [0,1]$, using $G_1(z)=0$, let us define $z(p)\triangleq  z\Lambda'(z;\vC)/\Lambda(z;\vC)$. Then proceeding as with the proof of Lemma~\ref{lem:delta-rho}, we see that $z(p)$ is monotone increasing with $z(0)=0$. Therefore, $z^*=z(p^*)$ is zero.
	
	Similarly, given $p^*$ and $x^*$, we use $G_3(x^*,y)=0$ to define $\delta(y) = y\Lambda_{y}(x^*,y;\vD)/\Lambda(x^*,y;\vD)$.
	Again, we can proceeding as with the proof of Lemma~\ref{lem:delta-rho} to show that $\delta(y)$ is monotone increasing. Also, since $\delta(y^*)<\delta_{\rm max}=\delta(1)$, we have that $y^*\in [0,1]$. 
\end{proof}
 
Therefore, to determine $\rmr(\delta)$ for any fixed $\delta$, it suffices to find 
$x$, $y$, $z$, $p$ such that $G_1(z) = G_2(x,y) = G_3(x,y) = 0$ and $x=z$. 

Now, the optimization in Theorem~\ref{thm:lagrange} does not constrain the values of $p$.
Furthermore, for certain constrained systems, there are instances where $p$ falls outside the interval $[0,1]$.
In this case, instead of solving the optimization problem \eqref{eq:MR}, we set $p$ to be either zero or one, and we solve the corresponding optimization problems~\eqref{eq:Sp} and~\eqref{eq:Tp}.
Specifically, if we have $p^* < 0$, then we set $p^* = 0$ and $x^* = 0$, or if $p^* > 1$, then we set $p^* = 1$ and $x^* = \infty$.
Hence, the resulting rates we obtain is a {\em lower bound} for the GV-MR bound.
\vspace{2mm}

\noindent{\bf Procedure 2 \Big($\rmr(\delta)$ for fixed $\delta \le \delta_{\rm max}$\Big)}.

\noindent{\sc Input}: 
Matrices $\vC_{\cG}(x)$, $\vD_{\cG}(x,y)$\\[1mm]
\noindent{\sc Output}: $\rmr(\delta)$ or $R_{\rm LB}(\delta)$, where $\rmr(\delta)\ge R_{\rm LB}(\delta)$.
\begin{enumerate}[(1)]
	\item Apply the Newton-Raphson method to obtain $p^*,x^*,y^*$ such that $G_1(x^*)$, $G_2(x^*,y^*)$ and $G_3(x^*,y^*)$ are approximately zero. 
	Specifically, we do the following.
	\begin{itemize}
		\item Fix some tolerance value $\epsilon$
		\item Set $t = 0$ and pick some initial guess $p_t\ge 0$, $x_t\ge 0$, $0 \le y_t \le 1$.
		\item While $|p_{t}-p_{t-1}|+ |x_{t}-x_{t-1}| +|y_{t}-y_{t-1}|>\epsilon$~:
		\begin{itemize}
			\item Compute the next guess $p_{t+1},x_{t+1},y_{t+1}$:
			\begin{equation*}
				\begin{bmatrix}
					p_{t+1}\\
					x_{t+1}\\
					y_{t+1} \\
				\end{bmatrix}
				=
				\begin{bmatrix}
					p_{t}\\
					x_{t}\\
					y_{t} \\
				\end{bmatrix}
				-
				\begin{bmatrix}
					\frac{\partial G_1}{\partial p} & 	\frac{\partial G_1}{\partial x} & 	\frac{\partial G_1}{\partial y}\\
					\frac{\partial G_2}{\partial p} & 	\frac{\partial G_2}{\partial x} & 	\frac{\partial G_2}{\partial y}\\
					\frac{\partial G_3}{\partial p} & 	\frac{\partial G_3}{\partial x} & 	\frac{\partial G_3}{\partial y}\\
				\end{bmatrix}^{-1}
				\begin{bmatrix}
					G_1(x_t)\\
					G_2(x_t,y_t)\\
					G_3(x_t,y_t) \\
				\end{bmatrix}\,.
			\end{equation*}
			\item Here, we apply the power iteration method to compute 
			$\Lambda(x_t;\vC)$, $\Lambda'(x_t;\vC)$, $ \Lambda''(x_t;\vC)$, $\Lambda(x_t,y_t;\vD)$,
			$\Lambda_x(x_t,y_t;\vD)$, $\Lambda_y(x_t,y_t;\vD)$, $\Lambda_{xx}(x_t,y_t;\vD)$, $\Lambda_{yy}(x_t,y_t;\vD)$, and $\Lambda_{xy}(x_t,y_t;\vD)$.
			\item Increment $t$ by one.
		\end{itemize}
		\item Set $p^*\gets p_t$, $x^*\gets x_t$, $y^*\gets y_t$.
	\end{itemize}

	\item[(2A)] If $ 0 \le p^* \le 1$, set {$\rmr(\delta) \gets  2\log \Lambda(x^*;\vC) + \delta \log y^* - \log \Lambda(x^*,y^*;\vD)$}.
	\item[(2B)] Otherwise,
		\begin{itemize}
			\item if $p^* < 0$, set $p^* \gets 0$, $x^* \gets 0$, and $y^* \gets \text{solution of }G_3(0,y)=0$.
			\item if $p^* > 1$, set $p^* \gets 1$, $x^* \gets \infty$, and $y^* \gets \text{solution of }G_3(\infty,y)=0$.
		\end{itemize}
		Finally, we set {$R_{\rm LB}(\delta) \gets  2\log \Lambda(x^*;\vC) + \delta \log y^* - \log \Lambda(x^*,y^*;\vD)$}.

\end{enumerate}	

\begin{remark}Let $p^*$ be the value computed at Step 1. When $p^*$ falls outside the interval $[0,1]$, we set $p^*\in\{0,1\}$ and we argued earlier that the value returned $R_{\rm LB}(\delta)$ (at Step 2B) is at most $\rmr(\delta)$. Nevertheless, we {\em conjecture} that $R_{\rm LB}(\delta)=\rmr(\delta)$.
\end{remark}

As before, we develop a plotting procedure that minimizes the use of Newton-Raphson iterations.

Note that we have three scenarios for $\Delta(x)$. If $\Delta(x)$ is monotone decreasing, then $\delta_{\max} = \lim_{x\to 0} \Delta(x)$ and we set $x^\#=0$.
If $\Delta(x)$ is monotone increasing, then $\delta_{\max} = \lim_{x\to \infty} \Delta(x)$ and we set $x^\#=\infty$.
Otherwise, $\Delta(x)$ is maximized for some positive value and we set $x^\#$ to be this value. Next, to obtain the GV-MR curve (see Remark~\ref{rem:curveGVMR}), we iterate over $x \in \interval{1,x^\#}$. Note that if $y(x^\#) < 1$ or equivalently $\delta(x^\#) < \delta_{\rm max}$, we obtain a lower bound on the GV-MR curve by iterating over $y \in \interval{y(x^\#),1}$. Similar to Theorem~\ref{thm:curve}, we define
 
\begin{equation}\label{eq:GV-MR}
\rho_{\rm MR}(x) \triangleq 2\log \Lambda(x;\vC) + \delta(x) \log y(x) - \log \Lambda(x,y(x);\vD)\,, 
\end{equation}

and

\begin{equation}\label{eq:GV-LB}
	\rho_{\rm LB}(y) \triangleq 2\log \Lambda(x^\#;\vC) + \delta(y) \log y - \log \Lambda(x^\#,y;\vD)\,.
\end{equation}	

Finally, we state the following analogue of Theorem~\ref{thm:curve}.

\begin{theorem}\label{thm:curve-MR}
	Define $\delta_{\max}$, $x^\#$ as before.
	For $x \in \interval{1,x^\#}$, we set
	\begin{align*}
		p(x) & \gets x\Lambda'(x;\vC)/\Lambda(x;\vC),\\
		y(x) & \gets \text{solution of $G_2(x,y)=0$},\\
		\delta(x) & \gets  y(x)\Lambda_y(x,y(x);\vD)/\Lambda(x,y(x);\vD),
	\end{align*}
	
If $y(x^\#) < 1$, then for $y \in \interval{y(x^\#),1}$, we set
	\begin{align*}
		\delta(y) & \gets  y\Lambda_y(x^\#,y;\vD)/\Lambda(x^\#,y;\vD)\,,
	\end{align*}

then the corresponding GV-MR curve is given by 
	\begin{equation}\label{eq:MRcurve}
		\big\{(\delta(x), \rho_{\rm MR}(x)) : x\in \interval{1,x^\#} \big\}
		\cup \{(\delta(y),\rho_{\rm LB}(y)) : y\in \interval{y(x^\#),1}\}
		\cup \big\{ (\delta, 0) : \delta\ge \delta_{\rm max} \big\}\,. 
	\end{equation}

where $\rho_{\rm MR}$ and $\rho_{\rm LB}$ are defined in \eqref{eq:GV-MR} and \eqref{eq:GV-LB} respectively.
\end{theorem}

\begin{example}
	We continue our example and evaluate the GV-MR bound for the $(3,2)$-SWCC constrained system. In this case, the matrices of interest are
	{
		
		\begin{equation*}
			\vC_{\cG} (z)= 
			\begin{bmatrix}
				z & 1 & 0\\
				0 & 0 & z\\
				z & 0 & 0
			\end{bmatrix}\text{ and }
			\vD_{\cG\times \cG} (x,y)= 
			\begin{bmatrix}
				x^2 & 2xy & 0 & 1 & 0 & 0\\
				0 & 0 & x^2 & 0 & xy & 0\\
				x^2 & xy & 0 & 0 & 0 & 0\\
				0 & 0 & 0 & 0 & 0 & x^2\\
				0 & 0 & x^2 & 0 & 0 & 0\\
				x^2 & 0 & 0 & 0 & 0 & 0 \\
			\end{bmatrix}\,.
		\end{equation*}
	}
	
	Here, we observe that $\Delta(x)$ is a monotone decreasing function
	and so, we set $x^\#=0.01$ and $\delta_{\max}=\lim_{x\to 0}\Delta(x)\approx 0.426$. If we apply Procedure 2 to compute $\rmr(\delta)$ for 100 points in $\interval{0,\delta_{\max}}$, we require $437$ Newton-Raphson iterations and $85500$ power iterations.
	In contrast, we use Theorem~\ref{thm:curve-MR} to  compute $(\delta(x),\rho_{\rm MR}(x))$ for 100 values of $x$ in the interval $\interval{1,x^\#}$.
	This requires $323$ Newton-Raphson iterations and involves $22296$ power iterations.
	The resulting GV-MR curve is given in Fig.~\ref{fig:swcc}(a).
	\hfill \IEEEQEDhere
\end{example}

\begin{remark}\label{rem:curveGVMR}
	Strictly speaking, the GV-MR curve described by~\eqref{eq:MRcurve} may not equal to the curve defined by the optimization problem~\eqref{eq:GV-MR}.
	Nevertheless, the curve provides a lower bound for the optimal asymptotic code rates and we {\em conjecture} that the GV-MR curve described by \eqref{eq:MRcurve} is a lower bound for the curve defined by the optimization problem~\eqref{eq:GV-MR}.
\end{remark}

	

\section{Single-State Graph Presentation }
\label{sec:ss}

In this section, we focus on graph presentations that have exactly one state. 
Here, we allow these single-state graph presentations to contain the parallel edges and 
their labels to be binary strings of length possibly greater than one.
Now, for these constrained systems, the procedures to evaluate the GV bound and its MR improvements can be greatly simplified.
This is because the matrices $\vB_{\cG\times \cG}(y)$, $\vC_{\cG}(z)$, $\vD_{\cG\times \cG}(x,y)$ are all of dimensions one by one. Therefore, determining their respective dominant eigenvalues is straightforward and does not require the power iteration method.
The results in this section follow directly from previous sections and our objective is to provide explicit formulas whenever possible. 

Formally, let $\cS$ be the constrained system with graph presentation $\cG = (\cV,\cE,\cL)$ such that $|V| = 1$ and $\cL: \cE \to \Sigma^{s}$ with $s\ge 1$. We further define $\alpha_t \triangleq \# \{(\vx,\vy) \in \cL(\cE)^2: d_{\vH}(\vx,\vy)=t\}$ for $0\le t\le s$.
Then the corresponding adjacency and reduced distance matrices are as follows:
\[ \vA_{\cG} = \begin{bmatrix}|\cE|\end{bmatrix} \text{ and }
	\vB_{\cG\times\cG}(y) = \begin{bmatrix} \sum_{t\ge 0} \alpha_t y^t \end{bmatrix}\,.
\]

Then we compute the capacity using its definition as $\capy(\cS) = (\log|\cE|)/s$.  

To compute $\rateT(\delta)$, we consider the following extension of the optimization problem \eqref{eq:dual} for the case $s\ge 1$. 
\begin{align}
	\rateT(\delta) 
	&= \frac{1}{s}\inf\left\{-\delta s\log y + \log \lambda(y;\vB) : 0\le y\le 1  \right\} \notag\\
	&= \frac{1}{s}\inf\left\{-\delta s\log y + \log\left(\sum _{t\ge 0}\alpha_ty^{t}\right) : 0\le y\le 1  \right\}. \label{eq:ssdual}
\end{align} 

As before, following the convexity of the objective function in \eqref{eq:ssdual}, we have that the optimal $y$ is the zero (in the interval $[0,1]$) of the function

\begin{equation}{\label{eq:ssderivative}}
	F(y) \triangleq \sum_{t \ge 0} (t-\delta s)\alpha_ty^t.
\end{equation}   

So, for fixed values of $\delta$, we can use the Newton-Raphson procedure to compute the root $y$ of \eqref{eq:ssderivative}, and hence, evaluate $\rgv(\delta)$. Note that the power iteration method is not required in this case.

On the other hand, to plot the GV curve, we have the following corollary of Theorem~\ref{thm:curve}.

\begin{corollary}\label{cor:sscurve}
	Let $\cG$ be the single-state graph presentation for a constrained system $\cS$.
	Then the corresponding GV curve is given by
	\begin{equation}\label{eq:GV-SECC} 
		\gvcurve(\cS) \triangleq \big\{(\delta, \rgv(\delta)) : \delta\in [0,1] \big\} 
		= \big\{(\delta(y), \rho(y)) : y\in [0,1] \big\} 
		\cup \big\{ (\delta, 0) : \delta\ge \delta_{\rm max} \big\}\,, 
	\end{equation}
	where
	\begin{align*}
		\delta_{\rm max} & = \frac{\sum_{t\ge 0}t\alpha_t }{s|\cE|^2}\,,\\
		\delta(y) &= \frac{\sum_{t \geq 0} t\alpha_ty^t}{s\left(\sum_{t \geq 0} \alpha_ty^t\right)}\,,\\
		\rho(y) & = \frac{1}{s}\left(\log\frac{|\cE|^2}{\sum_{t\ge 0}\alpha_t y^t} - \frac{\sum_{t \geq 0} t\alpha_ty^t}{\sum_{t \geq 0} \alpha_ty^t } \log y \right)\,. 
	\end{align*}
\end{corollary}

\vspace{1mm}

We illustrate this evaluation procedure via an example of the class of {\em subblock energy constrained codes (SECC)}. 
Formally, we fix a {\em subblock length} $L$ and {\em energy constraint} $w$. 
A binary word $\vx$ of length $mL$ said to satisfies the $(L,w)$-{\em subblock energy constraint} if 
we partition $\vx$ into $m$ subblocks of length $L$, then the number of ones in every subblock is at least $w$. 
We refer to the collection of words that meet this constraint as an $(L,w)$-{\em SECC constrained system}.
The class of SECCs was introduced by Tandon \etal{} for the
application of simultaneous energy and information transfer \cite{Tandon2016}. 
Later, in \cite{Tandon2018}, a GV-type bound was introduced (see \cite[Proposition~12]{Tandon2018} and also,~\eqref{eq:LB}) and we make comparisons with the GV bound~\eqref{eq:GV-SECC} in the following example.

\begin{example}\label{exa:gv-secc}
	Let $L=3$ and $w=2$ and we consider a $(3,2)$-SECC constrained system.
	It is straightforward to observe that the graph presentation is as follows with the single state {\tt x}.
	Here, $s=L=3$.
	
	\begin{center}
		\begin{tikzpicture}[scale=.8]
			\tikzstyle{every node} = [draw,shape = circle]
			\node (x) at (0, 0) {{\tt x}};
			\draw [->] (x) to [out=0, in=45, loop]  node[pos =0.25, right, draw = none] {{\tt 011}} (x);
			\draw [->] (x) to [out=90, in=135, loop]  node[pos =0.45, left, draw = none] {{\tt 101}} (x);
			\draw [->] (x) to [out=180, in=225, loop]  node[pos =0.45, left, draw = none] {{\tt 110}} (x);
			\draw [->] (x) to [out=270, in=315, loop]  node[pos =0.45, right, draw = none] {{\tt 111}} (x);	
		\end{tikzpicture}
	\end{center}
	
	Then the corresponding adjacency and reduced distance matrices are as follows:
	\[
	\vA_{\cG}  = 
	\begin{bmatrix}
		4
	\end{bmatrix},\,
	\vB_{\cG\times \cG} (y) = 
	\begin{bmatrix}
		4+6y+6y^2
	\end{bmatrix}\,.
	\]
	First we determine the GV bound at $\delta = 1/3$. Observe that $F(y)=-4+6y^2$ and so, the optimal point $y$ for \eqref{eq:ssdual} is $\sqrt{2/3}$ (the unique solution to $F(y)$ in the interval $[0,1]$).
	Hence, we have that $\rateT(1/3)\approx 1.327$. 
	On the hand, the capacity of $(3,2)$-SECC constrained system is $\capy(\cS) = 2/3$.
	Therefore, the GV bound is given by $\rgv(1/3)=0.006$. 
	
	In contrast, the GV-type lower bound given by \cite[Proposition~12]{Tandon2018} is zero.
	Hence, the evaluation of the GV bound yields a significantly better lower bound.
	
	To plot the GV curve, we observe that $\delta_{\rm max} =  3/8$ and thus, the curve is given by
	\begin{equation*} 
		\gvcurve(\cS) 
		= \left\{\left(\frac{y+2y^2}{2+3y+3y^2}, \frac{1}{3}\log \frac{8}{2+3y+3y^2}+\frac{3y+6y^2}{2+3y+3y^2}\log y \right) : y\in [0,1] \right\} 
		\cup \left\{ (\delta, 0) : \delta\ge \frac38 \right\}\,.
	\end{equation*}
	
	We plot the curve in Section~\ref{sec:numerical}.
	\hfill \IEEEQEDhere
\end{example}

\vspace{2mm}
Next, we evaluate the GV-MR bound. To this end, we consider some proper subset $\cP\subset \cE$ and define
\begin{align*}
\alpha_t &\triangleq \#\{(\vx,\vy) \in \cL(\cE)^2:d_{\vH}(\vx,\vy)=t,~\vx,\vy\in \cP \},\\
\beta_t &\triangleq \#\{(\vx,\vy) \in \cL(\cE): d_{\vH}(\vx,\vy)=t,~(\vx\in\cP,\vy\notin\cP) \text{ or } ~(\vx\notin\cP,\vy\in\cP)\},\\
\gamma_t &\triangleq \#\{(\vx,\vy) \in \cL(\cE): d_{\vH}(\vx,\vy)=t,~\vx,\vy\notin\cP \}.
\end{align*}

Then we consider the following matrices:
\[ \vC_{\cG}(z) = \begin{bmatrix}|\cE|-|\cP|+|\cP|z\end{bmatrix} \text{ and }
\vD_{\cG\times\cG}(x,y) = \begin{bmatrix} \sum _{t\ge 0}(\alpha_tx^2 + \beta_t x + \gamma_t)y^{t} \end{bmatrix}\,.
\]

Setting $p$ to be the normalized frequency of edges in $\cP$, we obtain $S(p)$ by solving the optimization problem \eqref{eq:Sp}.

Specifically, we have that 
\begin{equation}\label{eq:ss-Sp}
	S(p)=\frac1s\left(H(p)+ p +\log |\cP| + (1-p) \log (|\cE|-|\cP|)\right)\,,
\end{equation}
and this value is achieved when
\begin{equation}\label{eq:ss-zopt}
	z=\frac{p(|\cE|-|\cP|)}{(1-p)|\cP|}\,.
\end{equation}

To compute $\rateT(p,\delta)$, we consider the following extension of the optimization problem \eqref{eq:Tp} for the case $s\ge 1$. 

\begin{align}
	\rateT(p,\delta) 
	&= \frac{1}{s}\inf\left\{-2p \log x -\delta s\log y + \log \lambda(y;\vD) : 0\le y\le 1  \right\} \notag\\
	&= \frac{1}{s}\inf\left\{-2p \log x -\delta s\log y + \log\left(\sum _{t\ge 0}(\alpha_t x^2 + \beta_t x + \gamma_t )y^{t}\right) : 0\le y\le 1  \right\}. \label{eq:ssmrdual}
\end{align} 

As before, following the convexity of the objective function in \eqref{eq:ssmrdual}, we have that the optimal $x$ and $y$ are the zeroes (in the interval $[0,1]$) of the functions

\begin{align}{\label{eq:ssmrderivative}}
	G_2(x,y) \triangleq & 2(1-p)(\sum_{t \geq 0}{\alpha_ty^t})x^2+(1-2p)(\sum_{t \geq 0}{\beta_ty^t})x -2p(\sum_{t \geq 0}{\gamma_ty^t}) \notag \\
	G_3(x,y) \triangleq & \sum_{t \ge 0}(t-\delta s)(\alpha_tx^2 + \beta_t x +\gamma_t)y^t	
\end{align}   

So, for fixed values of $p$ and $\delta$, we can use the Newton-Raphson procedure to compute the roots $x$ and $y$ of \eqref{eq:ssmrderivative}, and hence, evaluate $\rgv(p,\delta)$. Note that the power iteration method is not required in this case. We find $x^\#$ as defined in Section~\ref{sec:mr} and set

\begin{equation}\label{eq:GV-MR-SS}
	\rho_{\rm MR}(x) \triangleq 2\log (|\cE|-|\cP|+|\cP|x) + \delta(x) \log y(x) - \log \sum _{t\ge 0}(\alpha_tx^2 + \beta_t x + \gamma_t)y(x)^{t}\,.
\end{equation}

Furthermore if $y(x^\#) < 1$, we set 

\begin{equation}\label{eq:GV-LB-SS}
\rho_{\rm LB}(y) \triangleq 2\log (|\cE|-|\cP|+|\cP|x^\#) + \delta(y) \log y - \log \sum _{t\ge 0}(\alpha_t(x^\#)^2 + \beta_t x^\# + \gamma_t)y^{t}\,.	
\end{equation}

Next, to plot the GV-MR curve, we have the following corollary of Theorem~\ref{thm:curve-MR}.

\begin{corollary}\label{cor:sscurve-MR}
	Let $\cG$ be the single-state graph presentation for a constrained system $\cS$.
For $x \in \interval{1,x^\#}$, we set
\begin{align*}
	p(x) &= \frac{|\cP|x}{(|\cE| - |\cP|) + |\cP|x)},\\
	\delta(x) &=  \frac{\sum _{t\ge 1}t(\alpha_t x^2 + \beta_t x + \gamma_t )y(x)^{t}}{s\sum _{t\ge 0}(\alpha_t x^2 + \beta_t x + \gamma_t )y(x)^{t}},
\end{align*}

where $y(x)$ is the smallest root of the equation

\begin{align*}
	2(|\cE|-|\cP|)(\sum_{t \geq 0}{\alpha_ty^t})x+(|\cE|-|\cP| - |\cP|x)(\sum_{t \geq 0}{\beta_ty^t}) -2|\cP|(\sum_{t \geq 0}{\gamma_ty^t}) &= 0.
\end{align*}

If $y(x^\#) < 1$, then for $y \in \interval{y(x^\#),1}$, we set
\begin{align*}
	\delta(y) &=  \frac{\sum _{t\ge 1}t(\alpha_t (x^\#)^2 + \beta_t x^\# + \gamma_t )y^{t}}{s\sum _{t\ge 0}(\alpha_t (x^\#)^2 + \beta_t x^\# + \gamma_t )y^{t}},
\end{align*}
 
Then the corresponding GV-MR curve is given by 
	\begin{equation}
		\big\{(\delta(x), \rho_{\rm MR}(x)) : x\in \interval{1,x^\#} \big\}
		\cup \{(\delta(y),\rho_{\rm LB}(y)) : y\in \interval{y(x^\#),1}\}
		\cup \big\{ (\delta, 0) : \delta\ge \delta_{\rm max} \big\}\,. 
	\end{equation}
where $\rho_{\rm MR}$ and $\rho_{\rm LB}$ are defined in \eqref{eq:GV-MR-SS} and \eqref{eq:GV-LB-SS} respectively.

\end{corollary}

\begin{example}\label{exa:gvmr-secc}
	We continue our example and evaluate the GV-MR bound for the $(3,2)$-SECC constrained system. We have the following single state graph presentation.
	
	\begin{center}
		\footnotesize
		\begin{tikzpicture}[scale=1.2, transform shape]
			\tikzstyle{every node} = [draw,shape = circle]
			\node (A) at (0, 0) {{\tt A}};
			\draw [->] (A) to [out=0, in=30, loop]  node[pos =0.2, right, draw = none] {{\tt 011}} (A);
			\draw [->] (A) to [out=120, in=150, loop]  node[pos =0.2, left, draw = none] {{\tt 101}} (A);
			\draw [->] (A) to [out=240, in=270, loop]  node[pos =0.2, left, draw = none] {{\tt 110}} (A);
		\end{tikzpicture} 
	\end{center}
	
	Then the matrices of interest are:
	\[
	\vC_{\cG}  = 
	\begin{bmatrix}
		1 + 3z
	\end{bmatrix},\,
	\vD_{\cG\times \cG} (x,y) = 
	\begin{bmatrix}
		(3+6y^2)x^2 + 6xy + 1
	\end{bmatrix}\,.
	\]
	
	Since $\vC_{\cG}$ and $\vD_{\cG\times \cG} (x,y)$ are both singleton matrices, we have $\Lambda(z;\vC) = 1+3z$ and $\Lambda(x,y;\vD) = (3+6y^2)x^2 + 6xy + 1$. Then, $G_1(z) = -p(1+3z) + 3z$, $G_2(x,y) = 3(1+2y^2)x^2(1-p) + 3xy(1-2p) - p$ and $G_3(x,y) = 4x^2y^2 -3\delta(1+2y^2)x^2 + 2xy(1-3\delta) - \delta$. Now we apply Theorem ~\ref{thm:lagrange} and express $p,y,\delta$ in terms of $x$ where $x \in [1,x^\#]$ where $x^\# \to \infty$. 
	\begin{align*}
		p &= \frac{3x}{(1+3x)} \\
		y &= \frac{x-1}{2x} \\
		\delta &= \frac{2x(x-1)}{(9x^2-1)} 
	\end{align*}   
	
	Now observe that we have $y(x^\#) = 1/2$. Since we can still increase $y$ to $1$, we apply the GV-bound with $p = 1$ and $x = z = x^\#$ once we reach at the boundary that is $p = 1$. Hence at the boundary, we solve the following problem.
	\begin{align*}
		S(1) &= 2\log 3 \\
		\rateT(1,\delta) &= \inf\big\{-2\log x -3\delta \log y +\log (3(1 + 2y^2)x^2 + 6xy + 1): 1/2\le y\le 1; x = x^\# \to \infty  \big\} \\ &= \inf\big\{-3\delta \log y +\log 3 + \log (1+2y^2): 1/2\le y\le 1 \big\}\\
		\rmr(\delta) &= S(1) - \rateT(1,\delta).
	\end{align*}  		
	
	By setting $F(y) = -3\delta(1+2y^2) + 4y^2 = 0$, we get $\delta = 4y^2/3(1+2y^2)$ where $y \in [1/2,1]$ and we plot the respective curve.
	\IEEEQEDhere
\end{example}

\begin{figure}[!t]
	\begin{enumerate}[(a)]
		\item Lower bounds for $R(\delta;\cS)$ where $\cS$ is the class of $(3,2)$-SWCC

		\begin{center}
			\includegraphics[width=7.5cm]{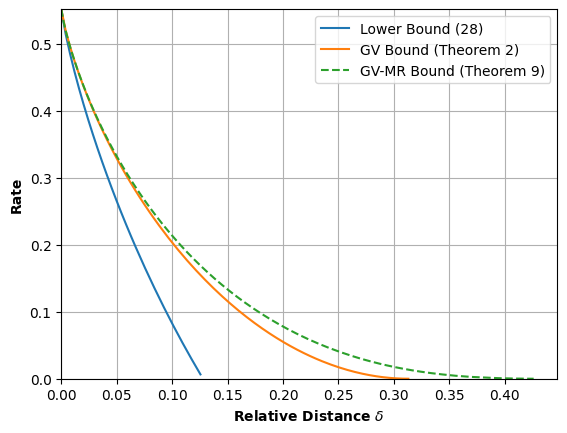}
		\end{center}
		
		\item  Lower bounds for $R(\delta;\cS)$ where $\cS$ is the class of $(10,7)$-SWCC
		
		\begin{center}
			\includegraphics[width=7.5cm]{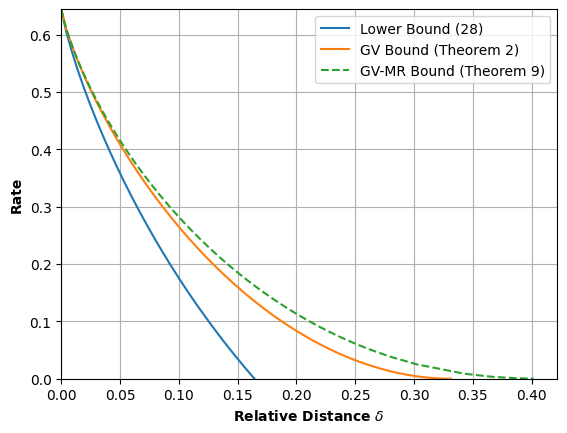}
		\end{center}
	\end{enumerate}
	\caption{Lower bounds for optimal asymptotic code rates $R(\delta;\cS)$ for the class of sliding-window constrained codes
	}
	\label{fig:swcc}
\end{figure}

\section{Numerical Plots}\label{sec:numerical}

In this section, we apply our numerical procedures to
compute the GV and the GV-MR bounds for some specific constrained systems. In particular, we consider the $(L,w)$-SWCC constrained systems defined in Section~\ref{sec:gv}, the ubiquitous $(d,k)$-runlength-limited systems (see for example~\cite[p3]{MRS2001}) and the $(L,w)$-subblock energy constrained codes recently introduced in~\cite{Tandon2016}.
In addition to the GV and GV-MR curves, we also plot a simple lower bound. 
For each $\delta \in \interval{0, 1/2}$, any ball size is at most $2\HH(\delta n)$. So, for any constrained system $\cS$, we have that $\widetilde{T}(\delta) \le \capy (\cS) + \HH(\delta)$. Therefore, we have that
\begin{equation}\label{eq:LB}
	R(\delta; \cS) \le \capy(\cS) - \HH(\delta)\,.	
\end{equation}
From the plots in Figures~\ref{fig:swcc},~\ref{fig:rll} and~\ref{fig:secc}, it is also clear that the computations of~\eqref{eq:gvcurve} and~\eqref{eq:MRcurve} yield a significantly better lower bound.

\subsection{$(L,w)$-Sliding Window Constrained Codes}

Fix $L$ and $w$. Recall from Section~\ref{sec:gv}, a binary word satisfies the $(L,w)$-{sliding window weight-constraint} if the number of ones in every consecutive $L$ bits is at least $w$ and the $(L,w)$-{SWCC constrained system} refers to the collection of words that meet this constraint.
From~\cite{Immink2020,Immink.2020}, we have a simple graph presentation that uses only $\binom{L}{w}$ states. 
To validate our methods, we choose $(L,w)\in\{(3,2), (10,7)\}$ and the corresponding graph presentations have $3$ and $120$ states, respectively.
Applying the plotting procedures described in Theorems~\ref{thm:curve} and~\ref{thm:curve-MR}, we obtain Figure~\ref{fig:swcc}.

\subsection{$(d,k)$-Runlength Limited Codes}

Next, we revisit the ubiquitous runlength constraint.
Fix $d$ and $k$. We say that a binary word satisfies the $(d,k)$-{\em RLL constraint} if each run of zeroes in the word has a length of at least $d$ and at most $k$ . Here, we allow the first and last runs of zeroes are allowed to have length less than $d$ . We refer to the collection of words that meet this constraint as a $(d,k)$-{\em RLL constrained system}. It is well known that $(d,k)$-{\em RLL constrained system} has the graph presentation with $k+1$ states (see for example~\cite{MRS2001}). 
Here, we choose $(d,k)\in\{(1,3), (3,7)\}$ to validate our methods and apply Theorems~\ref{thm:curve} and~\ref{thm:curve-MR} to obtain Figure~\ref{fig:rll}.
For $(d,k)=(3,7)$, we corroborate our results with that derived in~\cite{Winick1996}. Specifically, Winick and Yang determined the GV bound~\eqref{eq:GV} for the $(3,7)$-RLL constraint and remarked that the ``evaluation of the (GV-MR) bound required considerable computation'' for ``a small improvement''.
In the following table, we verify this statement. 
 

	\begin{center}
		\begin{tabular}{ |l|l|l| }
			\hline
		$\delta$	& GV-MR bound~\eqref{eq:GV-MR}& GV Bound~\cite{Winick1996} (see also, \eqref{eq:GV})\\
			\hline
			0 &	0.406 & 0.406 \\
			0.05 & 0.255 & 0.225 \\
			0.1 & 0.163 & 0.163 \\
			0.15 & 0.095 & 0.094 \\
			0.2	& 0.048 & 0.044 \\
			0.25 & 0.018 & 0.012  \\
			\hline
		\end{tabular}	
	\end{center}

\begin{figure}[!t]
	\begin{enumerate}[(a)]
		\item Lower bounds for $R(\delta;\cS)$ where $\cS$ is the class of $(1,3)$-RLL
		
		\begin{center}
			\includegraphics[width=7.5cm]{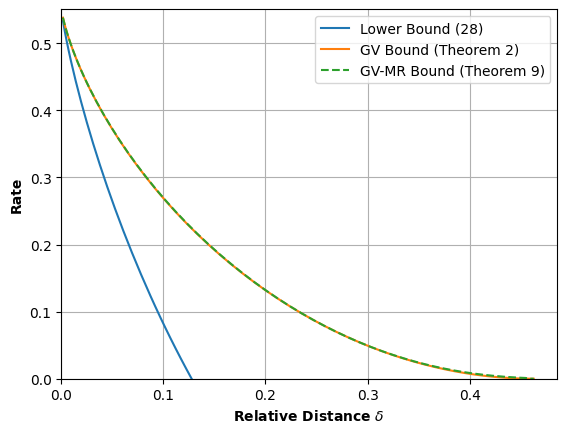}
		\end{center}
		
		\item  Lower bounds for $R(\delta;\cS)$ where $\cS$ is the class of $(3,7)$-RLL
		
		\begin{center}
			\includegraphics[width=7.5cm]{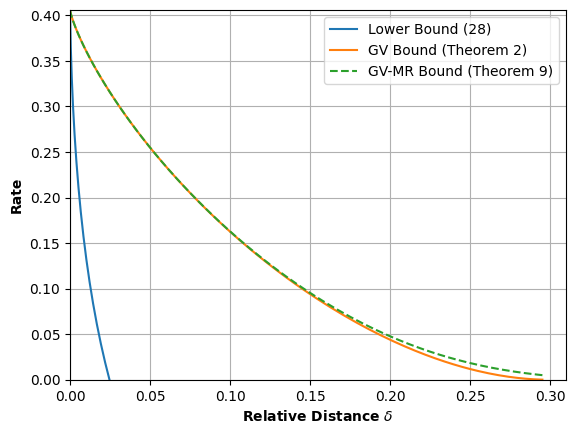}
		\end{center}                                                                                              
	\end{enumerate}
	\caption{Lower bounds for optimal asymptotic code rates $R(\delta;\cS)$ for the class of runlength limited  codes.
	}
	\label{fig:rll}
\end{figure}

\subsection{$(L,w)$-Subblock Energy Constrained Codes} 

Fix $L$ and $w$. 
Recall from Section~\ref{sec:ss}, a binary word satisfies the $(L,w)$-{subblock energy constraint} if each subblock of length $L$ have weight at least $w$
and the $(L,w)$-{SECC constrained system} refers to the collection of words that meet this constraint.
Then the corresponding graph presentation has a single state ${\tt x}$ with $\sum_{i=0}^{w}{L \choose i}$ edges, 
where each edge is labelled by a word of length $L$ and weight at least $w$. We apply the methods in Section~\ref{sec:ss} to determine the GV and GV-MR bounds.



For the GV bound, we provide the explicit formula for $\alpha_t$ and proceed as in Example~\ref{exa:gv-secc}. 
\begin{equation}
	\alpha_t =  {L \choose t}(|\cE| - \sum_{j=1}^{t}\sum_{k=0}^{\lceil \frac{j}{2} \rceil-1} {L-t \choose w-j+k} {t \choose k})
\end{equation}

Similarly, for GV-MR bound, we provide the explicit formula for $\alpha_t$, $\beta_t$ and $\gamma_t$ and proceed as in Example~\ref{exa:gvmr-secc}. 


\begin{align}
	\alpha_t &= {L \choose w}{L-w \choose i/2}{w \choose i/2} \;\text{if $t$ is even, otherwise $\alpha_t = 0$.} \\
	\beta_t &= 2{L \choose w}\sum_{j=1}^{\lfloor \frac{t}{2} \rfloor}{L-w \choose t-j}{w \choose j}-2\alpha_t \\
	\gamma_t &= {L \choose t}(|\cE| - \sum_{j=1}^{t}\sum_{k=0}^{\lceil \frac{j}{2} \rceil-1} {L-t \choose w-j+k} {t \choose k}) - \alpha_t - \beta_t
\end{align}

In Figure~\ref{fig:secc}, we plot the GV bound and GV-MR bounds. We remark that the simple lower bound~\eqref{eq:LB} corresponds to \cite[Proposition~12]{Tandon2018}.


\begin{figure}[!t]
	\begin{center}
	\includegraphics[width=7.5cm]{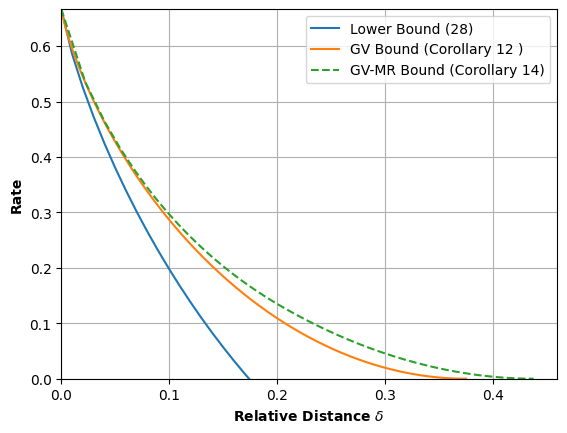}
	\end{center}
	\caption{Lower bounds for optimal asymptotic code rates $R(\delta;\cS)$ where $\cS$ is the class of $(3,2)$-SECC (subblock energy constrained codes).
	}
	\label{fig:secc}
\end{figure}

\newpage

\newpage

\appendices

\section{Power Iteration Method for Derivatives of Dominant Eigenvalues}
\label{app:power}
Throughout this appendix, we assume that $\vA$ is a diagonalizable matrix with dominant eigenvalue $\lambda_1$
and whose corresponding eigenspace has dimension one. Let $\ve_1$ be the unit eigenvector whose entries are positive in this space. 
Then the power iteration method is a well-known numerical procedure that finds the dominant eigenvalue $\lambda_1$ and the corresponding eigenvector $\ve_1$ efficiently. 

Now, in the preceding sections, the entries in the matrix $\vA$ are given functions in either one or two variables and thus, the dominant eigenvalue $\lambda_1$ is a function in the same variables.
Moreover, the numerical procedures in these sections require us to compute the higher order (partial) derivatives of this dominant eigenvalue function $\lambda_1$. 
To the best of our knowledge, we are unaware of any algorithms or numerical procedures that estimate the values of these derivatives.
Hence, in this appendix, we modify the power iteration method to compute these estimates.

Formally, let $\vA$ be an irreducible nonnegative diagonalizable square matrix with dominant eigenvalue $\lambda_1$ and corresponding unit eigenvector $\ve_1$. Since $\vA$ is diagonalizable, $\vA$ has $n$ eigenvectors $\ve_1,\ve_2, \ldots, \ve_n$ that form an orthonormal basis for $\RR^n$.
Let $\lambda_1,\lambda_2,\ldots, \lambda_n$ be the corresponding eigenvalues and so, we have that
\begin{equation}\label{eq:eigenvalue}
	\vA \ve_i=\lambda_i\ve_i \;\; \text{for all}\; i = 1,2,...,n.
\end{equation}
Since $\vA$ is irreducible, the dominant eigenspace has dimension one and also, the dominant eigenvalue is real and positive. Therefore, we can assume that $\lambda_1 > |\lambda_2| \ge \cdots \ge |\lambda_n|$.

We first assume that the entries of $\vA$ are functions in the variable $z$. Hence, $\lambda_i$'s and the entries of $\ve_i$'s are functions in $z$ too. Then Power Iteration~I then evaluates both $\lambda_1$ and $\lambda_1'$ for some fixed value of $z$, while Power Iteration~II additionally evaluates the second order derivative $\lambda_1''$.

The case where the entries of $\vA$ are functions in two variables $x$ and $y$ is discussed at the end of the appendix. Here, Power Iteration~III evaluates higher order partial derivatives of $\lambda_1$ for certain fixed values of $x$ and $y$. For ease of exposition, we provide detailed proofs for the correctness of Power Iteration~I and the proofs can be extended for Power Iteration~II and Power Iteration~III.   

We continue our discussion where the entries of $\vA$ are univariate functions in $z$. 
We differentiate each entry of $\vA$ with respect to $z$ to obtain the matrix $\vA'$. 
Also, for all $1\le i \le n$, we differentiate each entry of eigenvectors $\ve_i$ and the eigenvalue $\lambda_i$ to obtain $\ve_i'$ and $\lambda_i'$, respectively.
Specifically, it follows from \eqref{eq:eigenvalue} that
\begin{equation}\label{eq:eigenderivative}
	\vA' \ve_i+\vA \ve'_i=\lambda'_i\ve_i+\lambda_i\ve'_i \;\; \text{for all} \; i = 1,2,\cdots,n.
\end{equation}

Then the following procedure computes both $\lambda_1$ and $\lambda_1'$.

\vspace{2mm}

\noindent{\bf Power Iteration~I}.

\noindent{\sc Input}: 
Irreducible nonnegative diagonalizable matrix $\vA$\\[1mm]
\noindent{\sc Output}: Estimates of $\lambda_1$ and $\lambda_1'$ 
\begin{enumerate}[(1)]
	\item Intialize $\vq^{(0)}$ such that all its entries are strictly positive. 
	\begin{itemize}
		\item Fix some tolerance value $\epsilon$.
		\item While $|\vq^{(k)}-\vq^{(k-1)}|>\epsilon$~:
		\begin{itemize}
			\item Set
			\begin{align*}
				\lambda^{(k)} &= \|\vA \vq^{(k-1)}\|, \\
				\vq^{(k)}     &= \frac{\vA \vq^{(k-1)}}{\lambda^{(k)}}, \\
				\mu^{(k)}     &=\|\vA' \vq^{(k-1)} + \vA \vr^{(k-1)}-\lambda^{(k)}\vr^{(k-1)}\|,	\\
				\vr^{(k)}     &= \frac{\vA \vr^{(k-1)}+\vA' \vq^{(k-1)}-\mu^{(k)}\vq^{(k-1)}}{\lambda^{(k)}}.
			\end{align*}
			\item Increment $k$ by one.
		\end{itemize}
		\item Set $\lambda_1 \gets \lambda^{(k)}$ and $\lambda_1' \gets \mu^{(k)}$.
	\end{itemize}
\end{enumerate}

\begin{theorem}\label{thm:power1}
	If $\vA$ is irreducible nonnegative diagonalizable matrix and $\vq^{(0)}$ has positive components with unit norm, then as $k \to \infty$, we have 
	\[\lambda^{(k)} \to \lambda_1,~\vq^{(k)} \to \ve_1,~\mu^{(k)} \to \lambda_1'\,.\]
	Here, $\vq^{(k)} \to \ve_1$ means that $\left\|\vq^{(k)} - \ve_1\right\|\to 0$ as $k\to\infty$.
\end{theorem}

Before we present the proof of Theorem~\ref{thm:power1}, we remark that the usual power iteration method computes only $\lambda^{(k)}$ and $\vq^{(k)}$.
Then it is well-known (see for example, \cite{Stewart.1973}) that $\lambda^{(k)}$ and $\vq^{(k)}$ tends to $\lambda_1$ and $\ve_1$, respectively.

Now, since $\ve_i$'s span $\RR^n$, we can write $\vq^{(0)} = \sum_{i=1}^{n}\alpha_i\ve_i$ for any initial vector $\vq^{(0)}$.
The next technical lemma provides closed formulae for $\lambda^{(k)}$, $\vq^{(k)}$, $\mu^{(k)}$ and $\vr^{(k)}$ in terms of $\lambda_i$'s, $\ve_i$'s and $\alpha_i$'s.

\begin{lemma}\label{lem:power1}
	Let $\vq^{(0)} = \sum_{i=1}^{n}\alpha_i\ve_i$. Then,
	{\small 
		\begin{align} 
			\vq^{(k)} &= \frac{\sum_{i=1}^{n}\alpha_i\lambda_i^k\ve_i}{\|\sum_{i=1}^{n}\alpha_i\lambda_i^k\ve_i\|},  \label{eq:qk} \\
			\lambda^{(k)}&= \frac{\|\sum_{i=1}^{n}\alpha_i\lambda_i^k\ve_i\|}{\|\sum_{i=1}^{n}\alpha_i\lambda_i^{k-1}\ve_i\|}, \label{eq:lambk}\\
			\vr^{(k)} &= \frac{\sum_{i=1}^{n}(\alpha_i\ve'_i+\alpha_i'\ve_i)\lambda_i^{k}+(k\lambda_i'-\sum_{j=1}^{k}\mu^{(j)})\alpha_i\lambda_i^{k-1}\ve_i}{\|\sum_{i=1}^{n}\alpha_i\lambda_i^k\ve_i\|},  \label{eq:rk} \\
			\mu^{(k)}   & = \frac{\Bigg\|\sum_{i=1}^{n}(\alpha_i\ve'_i+\alpha_i'\ve_i)\lambda_i^{k-1}(\lambda_i-\lambda^{(k)}) + \alpha_i\lambda_i^{k-1}\lambda'_i\ve_i + ((k-1)\lambda_i'-\sum_{j=1}^{k-1}\mu^{(j)})\alpha_i\lambda_i^{k-2}(\lambda_i-\lambda^{(k)})\ve_i\Bigg\|}{\|\sum_{i=1}^{n}\alpha_i\lambda_i^{k-1}\ve_i\|}.\label{eq:muk}	
	\end{align}}
\end{lemma}

\begin{proof}
	Since $\vq^{(k)}$ is defined recursively as $\vq^{(k)} = \frac{\vA \vq^{(k-1)}}{\lambda^{(k)}} = \frac{\vA \vq^{(k-1)}}{\|\vA \vq^{(k-1)}\|}$, we have that 
	{\begin{equation*}
			\vq^{(k)} = \frac{\vA^k\vq^{(0)}}{\|\vA^k\vq^{(0)}\|}. 
	\end{equation*}  }
	Then it follows from Eq.~\eqref{eq:eigenvalue} that,
	{
		\begin{equation}
			\vA^k\vq^{(0)} = \vA^k\sum_{i=1}^{n}\alpha_i\ve_i   = \sum_{i=1}^{n}\alpha_i(\vA^k\ve_i)  = \sum_{i=1}^{n}\alpha_i\lambda_i^k\ve_i,
	\end{equation} }
	and so, we obtain \eqref{eq:qk}.
	Similarly, from \eqref{eq:eigenvalue}, we have that
	\begin{equation*}
		\lambda^{(k)} = \|\vA \vq^{(k-1)}\| 
		= \frac{\|\vA^k\vq^{(0)}\|}{\|\vA^{k-1}\vq^{(0)}\|} \\
		= \frac{\|\sum_{i=1}^{n}\alpha_i\lambda_i^k\ve_i\|}{\|\sum_{i=1}^{n}\alpha_i\lambda_i^{k-1}\ve_i\|},
	\end{equation*}
	as required for \eqref{eq:lambk}.
	
	Next, note that $\vr^{(0)} = \sum_{i=1}^{n}\alpha_i\ve'_i+\sum_{i=1}^{n}\alpha_i'\ve_i$.
	Then using the recursive definition of $\vr^{(k)}$, we have
	{	\begin{equation}
			\vr^{(k)} = \frac{\vA^{k}\vr^{(0)}+\sum_{j=0}^{k-1}\vA^j\vA'\vA^{k-j-1}\vq^{(0)}-(\sum_{j=1}^{k}\mu^{(j)})\vA^{k-1}\vq^{(0)}}{\|\vA^k\vq^{(0)}\|}.
	\end{equation} }
	Then from \eqref{eq:eigenvalue}, we have 
	{	\begin{equation}
			\vA^k\vr^{(0)} = \vA^k\left(\sum_{i=1}^{n}\alpha_i\ve'_i + \sum_{i=1}^{n}\alpha_i'\ve_i\right) =  \sum_{i=1}^{n}\alpha_i(\vA^k\ve'_i) + \sum_{i=1}^{n}\alpha_i'\lambda_i^k\ve_i.  
	\end{equation} }
	
	and from \eqref{eq:eigenderivative},
	\[
	\vA'\sum_{i=1}^{n}\alpha_i\lambda_i^{k-j-1}\ve_i = \sum_{i=1}^{n}\alpha_i\lambda_i^{k-j-1}(\vA'\ve_i) = \sum_{i=1}^{n}\alpha_i\lambda_i^{k-j-1}(\lambda'_i\ve_i+\lambda_i\ve_i'-\vA \ve_i'). 		
	\]
	
	Therefore, using \eqref{eq:eigenvalue} again,
	{\begin{align*}
			\sum_{j=0}^{k-1}\vA^j\vA'\sum_{i=1}^{n}\alpha_i\lambda_i^{k-j-1}e_i  &= \sum_{j=0}^{k-1}\vA^j\sum_{i=1}^{n}\alpha_i\lambda_i^{k-j-1}(\lambda'_ie_i+\lambda_ie_i'-\vA e_i') \\
			&= k\sum_{i=1}^{n}\alpha_i\lambda_i^{k-1}\lambda_i'e_i+\sum_{i=1}^{n}\alpha_i\lambda_i^{k}e'_i -\sum_{i=1}^{n}\alpha_i(\vA^ke'_i). 	
	\end{align*} }
	Therefore, we obtain \eqref{eq:rk}.
	
	Finally, recall that $\mu^{(k)}$ is defined as
	\begin{equation*}
		\mu^{(k)} =\|\vA' \vq^{(k-1)} + \vA \vr^{(k-1)} - \lambda^{(k)}\vr^{(k-1)}\|.
	\end{equation*}
	Then by replacing $\vr^{(k-1)}$ and $\vq^{(k-1)}$ from  \eqref{eq:rk} and \eqref{eq:qk}, respectively, and then using equation \eqref{eq:eigenderivative}, we obtain \eqref{eq:muk}.
\end{proof}

Finally, we are ready to demonstrate the correctness of Power Iteration~I.

\begin{proof}[Proof of Theorem~\ref{thm:power1}]
	Since $\vA$ is irreducible nonnegative diagonalizable matrix, $\lambda_1$ is real positive and there exists $0 <\epsilon < 1$ such that $\frac{|\lambda_i|}{\lambda_1} < \epsilon \;\; \text{for all} \; i = 2,3,\cdots,n$ (see, for example, \cite{MRS2001}). 
	For purposes of brevity, we write
	\begin{equation}\label{eq:Phik}
		\Phi_k = \sum_{i=1}^{n}\alpha_i\lambda_i^k\ve_i
	\end{equation}
	and so, we can rewrite \eqref{eq:qk} as 
	\[\vq^{(k)} = \frac{\Phi_k}{\|\Phi_k\|} = \frac{\lambda_1^k}{\|\Phi_k\|} \frac{\Phi_k}{\lambda_1^k} =  \frac{\lambda_1^k}{\|\Phi_k\|}\left(\alpha_1\ve_1+\sum_{i=2}^{n}\alpha_i\frac{\lambda_i^k}{\lambda_1^k}\ve_i\right)\,.\]
	
	Now, since ${\lambda_i^k}/{\lambda_1^k}\le \epsilon^k$ for all $i=2,\ldots, n$, we have that:
	\begin{equation}\label{eq:Phik-est}
		\left\|\frac{\Phi_k}{\lambda_1^k}-\alpha_1\ve_1\right\|\le C_1\epsilon^{k} \text{ for some constant } C_1.
	\end{equation}
	
	Then using the triangle inequality, we have that as $k\to \infty$, $\left|\frac{\|\Phi_k\|}{\lambda_1^k}-\alpha_1\right|\to 0$ and thus, 
	$\frac{\lambda_1^k}{\|\Phi_k\|}\to\frac{1}{\alpha_1}$.
	Therefore, $\|\vq^{(k)}-\ve_1\|\to 0$ as required.
	
	Note that since $\frac{\lambda_1^k}{\|\Phi_k\|}$ tends to a finite limit, we have that $\frac{\lambda_1^k}{\|\Phi_k\|}$ bounded above by some constant. In other words, we have that 
	\begin{equation}\label{normPhik}
		\frac{\lambda_1^k}{\|\Phi_k\|} \le C_2  \text{ for some constant } C_2.
	\end{equation}
	
	Next, we show the following inequality:
	\begin{equation}\label{eq:lambk-est}
		|\lambda^{(k)}-\lambda_1 |\le C_3\epsilon^{k-1} \text{ for some constant } C_3.
	\end{equation}
	Using \eqref{eq:lambk}, we have that 
	\[
	\frac{\|\Phi_k-\lambda_1\Phi_{k-1}\|}{\|\Phi_{k-1}\|}
	= \frac{\lambda_1^{k-1}}{\|\Phi_{k-1}\|} \frac{\sum_{i=1}^{n}\alpha_i\lambda_i^k\ve_i - \alpha_i\lambda_1\lambda_i^{k-1}\ve_i}{\lambda_1^{k-1}}
	= \left(\frac{\lambda_1^{k-1}}{\|\Phi_{k-1}\|}\right)\cdot \lambda_1 \cdot \sum_{i=2}^{n}\alpha_i\left(\frac{\lambda_i^k}{\lambda_1^k}-\frac{\lambda_i^{k-1}}{\lambda_1^{k-1}}\right)\ve_i\,.
	\]
	Now, observe that $\left(\frac{\lambda_i^k}{\lambda_1^k}-\frac{\lambda_i^{k-1}}{\lambda_1^{k-1}}\right)\le 2\epsilon^{k-1}$ for $i=2,\ldots, n$. Since $\frac{\lambda_1^{k-1}}{\|\Phi_{k-1}\|}\le C_2$, we have \eqref{eq:lambk-est} after applying the triangle inequality.
	
	Again, to reduce clutter, we introduce the following abbreviations.
	
	\begin{align*}
		D_k &= \sum_{i=1}^{n}(\alpha_i\ve_i'+\alpha_i'\ve_i)\lambda_i^{k-1}(\lambda_i-\lambda^{(k)}),\\
		E_k &= \sum_{i=1}^{n} \alpha_i\lambda_i^{k-1}\lambda'_i\ve_i, \\ 
		F_k &=  \sum_{i=1}^{n}\left((k-1)\lambda_i'-\sum_{j=1}^{k-1}\mu^{(j)}\right)\alpha_i\lambda_i^{k-2}(\lambda_i-\lambda^{(k)})\ve_i.
	\end{align*}
	Thus, we can rewrite \eqref{eq:muk} as 
	\[ \mu^{(k)} = \frac{\|D_k+E_k+F_k\|}{\|\Phi_{k-1}\|} \le \lambda_1' + \frac{\|D_k\|}{\|\Phi_{k-1}\|}+\frac{\|E_k-\lambda_1'\Phi_{k-1}\|}{\|\Phi_{k-1}\|}+\frac{\|F_k\|}{\|\Phi_{k-1}\|}\,.\]
	Next, we bound each of the summands on the right-hand side.
	Specifically, we show the following inequalities:
	\begin{align}
		\frac{\|D_k\|}{\|\Phi_{k-1}\|} + \frac{\|E_k-\lambda_1'\Phi_{k-1}\|}{\|\Phi_{k-1}\|}& \le C_4\epsilon^{k-1}  \text{ for some constant }C_4, \label{eq:Dk-est}\\
		\frac{\|F_k\|}{\|\Phi_{k-1}\|} & \le C_5(k-1)\epsilon^{k-1} + C_5\left(\sum_{j=1}^{k-1}\mu^{(k)} \right)\epsilon^{k-1}\text{ for some constant }C_5 \label{eq:Fk-est}.
	\end{align}
	
	To demonstrate \eqref{eq:Dk-est}, we consider 
	\[
	\frac{\|D_k\|}{\lambda_1^{k-1}} = 
	\left\|\sum_{i=1}^n (\alpha_i\ve_i'+\alpha_i'\ve_i)\frac{\lambda_i^{k-1}}{\lambda_1^{k-1}}(\lambda_i-\lambda^{(k)})\right\|
	\le \|\alpha_1\ve_1'+\alpha_1'\ve_1\||\lambda_1-\lambda^{(k)}| + \epsilon^{k-1} \sum_{i=2}^n\| \alpha_i\ve_i'+\alpha_i'\ve_i\||\lambda_i-\lambda^{(k)}|.
	\]
	We use \eqref{eq:lambk-est} to bound the first summand by some constant multiple of $\epsilon^{k-1}$.
	On the other hand, we have $|\lambda_i-\lambda^{(k)}| \leq |\lambda_i-\lambda_1| + |\lambda_1-\lambda^{(k)}| \le  \max\{|\lambda_i-\lambda_1|: 2 \leq i \leq n\} + C_3\epsilon^{k-1}$ for $2 \leq i \leq n$. In other words, the second summand is also bounded by some constant multiple of $\epsilon^{k-1}$.  
	Next, we consider
	\[
	\frac{\|E_k-\lambda_1'\Phi_{k-1}\|}{\lambda_1^{k-1}} = 
	\left\|\sum_{i=1}^n \alpha_i\frac{\lambda_i^{k-1}}{\lambda_1^{k-1}}(\lambda_i'-\lambda_1')\ve_i\right\|
	\le  \epsilon^{k-1} \sum_{i=2}^n |\alpha_i(\lambda'_i-\lambda'_1)|.
	\]
	and so, $\frac{\|E_k-\lambda_1'\Phi_{k-1}\|}{\lambda_1^{k-1}}$ is also bounded by a multiple of $\epsilon^{k-1}$.  
	Therefore, since $\frac{\lambda_1^{k-1}}{\|\Phi_{k-1}\|}\le C_2$, we have \eqref{eq:Dk-est}.
	Using similar methods, we can establish \eqref{eq:Fk-est}.
	
	Next, we apply \eqref{eq:Dk-est} and then recursively apply \eqref{eq:Fk-est} until right hand side is free of $\mu^{(i)}$'s. Then it follows that,
	\begin{equation}\label{eq:power3}
		\footnotesize 
		\mu^{(k)} \le \lambda_1' + C_4\epsilon^{k-1} + C_5(k-1)\epsilon^{k-1}  
		+  \prod_{j=2}^{k-1}(1+C_5\epsilon^{k-j}) + C_5\epsilon^{k-1}\sum_{i=1}^{k-1}(\lambda_1' + C_4\epsilon^{k-i-1}   C_5(k-i-1)\epsilon^{k-i-1})\prod_{j=2}^{i}(1+C_{5}\epsilon^{k-j})). 
	\end{equation}
	
	Furthermore, since $i \le  k-1$, $\prod_{j=2}^{i}(1+C_5\epsilon^{k-j}) \leq \prod_{j=2}^{k-1}(1+C_5\epsilon^{k-j})$, 
	we can rewrite \eqref{eq:power3} as
	
	\begin{equation}\label{eq:power4}
		\small
		\mu^{(k)} \le \lambda_1' + C_4\epsilon^{k-1} + C_5(k-1)\epsilon^{k-1}  
		+  \prod_{j=2}^{k-1}(1+C_5\epsilon^{k-j})\left(1 + C_5\epsilon^{k-1}\sum_{i=1}^{k-1}(\lambda_1' + C_4\epsilon^{k-i-1}   C_5(k-i-1)\epsilon^{k-i-1})\right). 
	\end{equation}
	
	Next it follows from standard calculus that $\prod_{j=2}^{k-1}(1+C_{5}\epsilon^{k-j}) < e^{\frac{C_{5}}{1-\epsilon}}$. 
	Furthermore, since $\epsilon < 1$, we have $\sum_{i=0}^{k-2}\epsilon^j < \frac{1}{1-\epsilon}$ and $\sum_{i=0}^{k-2}j\epsilon^j < \frac{1}{(1-\epsilon)^2}$. Putting everything together, we have
	
	\begin{equation}\label{eq:power5}
		\mu^{(k)} \le  \lambda_1' + C_4\epsilon^{k-1} + C_5(k-1)\epsilon^{k-1} 
		+ C_{5}\epsilon^{k-1}e^{\frac{C_{5}}{1-\epsilon}}\left(1 + (k-1)\lambda_1' + \frac{C_{4}}{1-\epsilon} + \frac{C_{5}}{(1-\epsilon)^{2}}\right).
	\end{equation} 
	
	As $k \to \infty$, since $\epsilon < 1$, we have $\epsilon^{k} \to 0$ and $k\epsilon^{k} \to 0$.
	Therefore, $\lim_{k\to\infty} \mu^{(k)} \le \lambda_1'$. Using similar methods, we have that  $\lim_{k\to\infty} \mu^{(k)} \ge \lambda_1'$ and so, $\lim_{k\to\infty} \mu^{(k)}=\lambda_1'$, as required.
\end{proof}

\vspace{2mm}

Next, we modify Power Iteration~I so as to compute the higher order derivatives.
We omit a detailed proof as it is similar to the proof of Theorem~\ref{thm:power1}. 

\noindent{\bf Power Iteration~II}.

\noindent{\sc Input}: 
Irreducible nonnegative diagonalizable matrix $\vA$\\[1mm]
\noindent{\sc Output}: Estimates of $\lambda_1$, $\lambda_1'$ and $\lambda_1''$ 
\begin{enumerate}[(1)]
	\item Intialize $\vq^{(0)}$ such that all its entries are strictly positive. 
	\begin{itemize}
		\item Fix some tolerance value $\epsilon$.
		\item While $|\vq_{(k)}-\vq_{(k-1)}|>\epsilon$~:
		\begin{itemize}
			\item Set
			\begin{align*}
				\lambda^{(k)} &= \|\vA \vq^{(k-1)}\|, \\
				\vq^{(k)}     &= \frac{\vA \vq^{(k-1)}}{\lambda^{(k)}}, \\
				\mu^{(k)}     &=\|\vA' \vq^{(k-1)} + \vA \vr^{(k-1)}-\lambda^{(k)}\vr^{(k-1)}\|,	\\
				\vr^{(k)}     &= \frac{\vA \vr^{(k-1)}+\vA' \vq^{(k-1)}-\mu^{(k)}\vq^{(k-1)}}{\lambda^{(k)}}, \\
				\nu^{(k)} &= \|\vA''q^{(k-1)} + 2\vA'\vr^{(k-1)}+\vA s^{(k-1)} -\lambda^{(k)} \vs^{(k-1)}-2\mu^{(k)}\vr^{(k-1)}\|,\\
				\vs^{(k)} &=\frac{\vA''q^{(k-1)}+2\vA'\vr^{(k-1)}+\vA \vs^{(k-1)}-2\mu^{(k)}\vr^{(k-1)}-\nu^{(k)}\vq^{(k-1)}}{\lambda^{(k)}}.
			\end{align*}
			\item Increment $k$ by one.
		\end{itemize}
		\item Set $\lambda_1 \gets \lambda^{(k)}$, $\lambda_1' \gets \mu^{(k)}$ and $\lambda_1'' \gets \nu^{(k)}$ .
	\end{itemize}
\end{enumerate}

\begin{theorem}\label{thm:power2}
	If $\vA$ is irreducible nonnegative diagonalizable matrix and $\vq^{(0)}$ has positive components with unit norm, then as $k \to \infty$, we have 
	\[\lambda^{(k)} \to \lambda_1,~\vq^{(k)} \to \ve_1,~\mu^{(k)} \to \lambda_1', ~\nu^{(k)} \to \lambda_1''\,.\]
\end{theorem}

Finally, we end this appendix with a power iteration method that computes the partial derivatives when the elements of the given matrix are bivariate functions.
\vspace{1mm}

\noindent{\bf Power Iteration~III}.

\noindent{\sc Input}: 
Irreducible nonnegative diagonalizable matrix $\vA$\\[1mm]
\noindent{\sc Output}: Estimates of $\lambda_1$, $(\lambda_1)_x$, $(\lambda_1)_y$, $(\lambda_1)_{xx}$, $(\lambda_1)_{yy}$, $(\lambda_1)_{xy}$  
\begin{enumerate}[(1)]
	\item Intialize $\vq^{(0)}$ such that all its entries are strictly positive. 
	\begin{itemize}
		\item Fix some tolerance value $\epsilon$.
		\item While $|\vq^{(k)}-\vq^{(k-1)}|>\epsilon$~:
		\begin{itemize}
			\item Set
			{\small \begin{align*}
					\lambda^{(k)} & =\|\vA q^{(k-1)}\|,\\
					q^{(k)} &=\frac{\vA q^{(k-1)}}{\lambda^{(k)}},\\
					\lambda^{(k)}_x &= \|\vA_x \vq^{(k-1)}+\vA \vq^{(k-1)}_x-\lambda \vq^{(k-1)}_x\|,\\
					\vq^{(k)}_x&=\frac{\vA_x \vq^{(k-1)}+\vA \vq^{(k-1)}_x-\lambda^{(k-1)}_x \vq^{(k-1)}}{\lambda^{(k)}},\\
					\lambda^{(k)}_y &= \|\vA_y \vq^{(k-1)}+\vA \vq^{(k-1)}_y-\lambda \vq^{(k-1)}_y\|,\\
					\vq^{(k)}_y&=\frac{\vA_y \vq^{(k-1)}+\vA \vq^{(k-1)}_y-\lambda^{(k-1)}_y\vq^{(k-1)}}{\lambda^{(k)}},\\
					\lambda^{(k)}_{xx} &= \|\vA_{xx}\vq^{(k-1)}+2\vA_x\vq^{(k-1)}_x+\vA \vq^{(k-1)}_{xx}  -\lambda^{(k-1)}\vq^{(k-1)}_{xx}-2\lambda^{(k-1)}_{x}\vq^{(k-1)}_{x}\|,\\
					\vq^{(k)}_{xx}&=\frac{\vA_{xx}\vq^{(k-1)}+2\vA_x\vq^{(k-1)}_x + \vA \vq^{(k-1)}_{xx}-2\lambda^{(k-1)}_x\vq^{(k-1)}_x-\lambda^{(k-1)}_{xx}\vq^{(k-1)}}{\lambda^{(k)}} \\
					\lambda^{(k)}_{yy} &= \|\vA_{yy}\vq^{(k-1)}+2\vA_y\vq^{(k-1)}_y+\vA \vq^{(k-1)}_{yy}  -\lambda^{(k-1)}\vq^{(k-1)}_{yy}-2\lambda^{(k-1)}_{y}\vq^{(k-1)}_{y}\|,\\
					\vq^{(k)}_{yy}&=\frac{\vA_{yy}\vq^{(k-1)}+2\vA_y\vq^{(k-1)}_y+\vA \vq^{(k-1)}_{yy}-2\lambda^{(k-1)}_y\vq^{(k-1)}_y-\lambda^{(k-1)}_{yy}\vq^{(k-1)}}{\lambda^{(k)}}  \\
					\lambda^{(k)}_{xy} &= \|\vA_{xy}q^{(k-1)}+\vA_xq^{(k-1)}_y+\vA_y\vq^{(k-1)}_x +\vA \vq^{(k-1)}_{xy} 
					-\lambda^{(k-1)}\vq^{(k-1)}_{xy}-\lambda^{(k-1)}_{x}\vq^{(k-1)}_{y}-\lambda^{(k-1)}_{y}\vq^{(k-1)}_{x}\|,\\
					\vq^{(k)}_{xy} &= \frac{\vA_{xy}\vq^{(k-1)}+\vA_x\vq^{(k-1)}_y+\vA_y\vq^{(k-1)}_x+\vA \vq^{(k-1)}_{xy}-\lambda^{(k-1)}_{xy} \vq^{(k-1)}-\lambda^{(k-1)}_{x}\vq^{(k-1)}_{y}-\lambda^{(k-1)}_{y}\vq^{(k-1)}_{x}}{\lambda^{(k)}}.
			\end{align*} }
			\item Increment $k$ by one.
		\end{itemize}
		\item Set $\lambda^{(k)} \gets \lambda_1$, $\lambda^{(k)}_x \gets (\lambda_1)_x$, $\lambda^{(k)}_y\gets  (\lambda_1)_y$,
		$\lambda^{(k)}_{xx}\gets(\lambda_1)_{xx}$, $\lambda^{(k)}_{yy}\gets (\lambda_1)_{yy}$, $\lambda^{(k)}_{xy}\gets(\lambda_1)_{xy}$. 
	\end{itemize}
\end{enumerate}

\begin{theorem}\label{thm:power3}
	If $\vA$ is irreducible nonnegative diagonalizable matrix and $\vq^{(0)}$ has positive components with unit norm, then as $k \to \infty$, we have 
	$\lambda^{(k)}_{xx}\to(\lambda_1)_{xx}$, $\lambda^{(k)}_{yy}\to (\lambda_1)_{yy}$, $\lambda^{(k)}_{xy}\to(\lambda_1)_{xy}$\,.
\end{theorem}

\end{document}